\documentclass[authoryear,orivec,cleveref,thm-restate]{lipics-v2021}
\usepackage{graphicx} 
\usepackage{lineno}
\usepackage{url}
\usepackage{float}
\usepackage{amsmath}
\usepackage{amssymb}
\usepackage{amsthm}
\usepackage{algorithm}
\usepackage{algorithmicx}
\usepackage[noend]{algpseudocode}
\usepackage[pdftex,dvipsnames]{xcolor}
\usepackage{tikzlings}
\usepackage{orcidlink}
\usepackage{tikzducks}
\usepackage{tikzsymbols}
\usepackage{tikz}
\usepackage{mathtools}
\usepackage{thmtools}
\usepackage{hyperref}
\usepackage{xspace}
\usepackage[new]{old-arrows}
\usepackage{float}
\usepackage{complexity}
\usepackage{caption}
\usepackage{subcaption}
\usepackage{graphicx}
\usepackage{mathabx}
\usepackage{stmaryrd}
\usepackage{wasysym}
\usepackage{pgfplots}
\pgfplotsset{compat=1.17}

\newcommand{\subp}[1]{\subparagraph*{#1.}}

\algdef{SE}[DOWHILE]{Do}{doWhile}{\algorithmicdo}[1]{\algorithmicwhile\ #1}%

\newcommand{\Fc}{\mathcal{F}}

\newcommand{\Dc}{\mathcal{D}}

\newcommand{\Pc}{\mathcal{P}}

\newcommand{\prob}{\mathbb{P}}

\newcommand{\Nb}{\mathbb{N}}
\newcommand{\Rb}{\mathbb{R}}
\newcommand{\Eb}{\mathbb{E}}
\newcommand{\Qb}{\mathbb{Q}}
\newcommand{\Zb}{\mathbb{Z}}

\newclass{\TOWER}{TOWER}
\newclass{\ACKERMANN}{ACKERMANN}
\newclass{\EXPTIME}{EXPTIME}
\newclass{\IIEXPTIME}{2-EXPTIME}
\newclass{\NEXPTIME}{NEXPTIME}
\newclass{\SQRTSUM}{SqrtSum}
\newclass{\THREESAT}{3SAT}
\newclass{\PTIME}{PTIME}

\usepackage{graphicx}
\usepackage{array}
\usepackage{tikz}
\usetikzlibrary{automata, arrows, positioning, decorations.markings, decorations.pathreplacing}
\usetikzlibrary {decorations.pathmorphing, decorations.shapes}
\usetikzlibrary{backgrounds,fit,shapes.geometric}
\usetikzlibrary {graphs, quotes}
\usetikzlibrary{calc}
\usetikzlibrary{shapes,shapes.geometric,fit}
\usepackage[new]{old-arrows}
\usetikzlibrary {graphs,shapes.geometric,quotes}
\usetikzlibrary{decorations.pathreplacing}
\usetikzlibrary{automata}
\usetikzlibrary{calc}
\usetikzlibrary{arrows,automata,decorations.pathmorphing}
\usetikzlibrary{shapes,shapes.geometric,fit,calc,automata}
\usepackage{xargs} 
\usepackage{todonotes}
\newcommandx{\theju}[2][1=]{\todo[linecolor=blue,backgroundcolor=blue!25,bordercolor=blue,#1]{\tiny T: #2}}
\newcommandx{\leon}[2][1=]{\todo[linecolor=red,backgroundcolor=red!25,bordercolor=red,#1]{\tiny L:#2}}
\newcommandx{\aliA}[2][1=]{\todo[linecolor=green,backgroundcolor=red!25,bordercolor=red,#1]{\tiny L:#2}}

\usepackage{latexsym}
\usepackage{wasysym}

\newtheorem{problem}{Problem}

\theoremstyle{claimstyle}

\renewcommand{\Game}{\mathcal{G}}

\newcommand{\MDProc}{\mathcal{P}}
\newcommand{\MChain}{\mathcal{M}}

\newcommand{\bit}{\mathsf{bit}}
\newcommand{\rel}{\mathsf{rel}}

\newcommand{\distribution}{\mu}

\newcommand{\bsigma}{{\bar{\sigma}}}
\newcommand{\btau}{{\bar{\tau}}}

\newcommand{\bx}{{\bar{x}}}
\newcommand{\by}{{\bar{y}}}

\renewcommand{\p}{\mathsf{p}}
\newcommand{\val}{\mathsf{val}}

\newcommand{\mul}{\mathsf{mul}}
\newcommand{\sfdef}{\mathsf{def}}

\newcommand{\defas}{\coloneqq}

\renewcommand{\epsilon}{\varepsilon}
\renewcommand{\phi}{\varphi}

\newcommand{\solver}{\mathfrak{S}}

\renewcommand{\Circle}{\mathbin{\raisebox{-0.34ex}{\scalebox{1.5}{$\circ$}}}}
\newcommand{\smalltriangle}{\scalebox{0.5}{$\triangle$}}
\newcommand{\smallpentagon}{\scalebox{0.75}{$\pentagon$}}

\tikzset{every node/.append style={initial text={}, node distance=15mm}}

\tikzset{stoch/.style={state, fill=black, inner sep=0.5pt, text=white, minimum size=0pt}}
\tikzset{circlev/.style={state,minimum size=20pt,fill, top color=white, bottom color=black!20}}
\tikzset{squarev/.style={state, rectangle, minimum size=20pt,fill, top color=white, bottom color=black!20}}
\tikzset{diamondv/.style={state, diamond, minimum size=20pt,fill, top color=white, bottom color=black!20}}
\tikzset{trianglev/.style={state, regular polygon, regular polygon sides=3, minimum size=15pt, inner sep=0.8pt, fill, top color=white, bottom color=black!20}}
\tikzset{pentagonv/.style={state, regular polygon, regular polygon sides=5, minimum size=20pt,fill, top color=white, bottom color=black!20}}
\tikzset{hexagonv/.style={state, regular polygon, regular polygon sides=6, minimum size=20pt,fill, top color=white, bottom color=black!20}}
\tikzset{gadget/.style={state, rectangle, rounded corners, dashed}}
\tikzset{mem/.style={state, minimum size=25pt, text=white, fill, top color=black!40, bottom color=black!60, rectangle, rounded corners}}
\hideLIPIcs
\nolinenumbers
\authorrunning{Ali Asadi, L\'eonard Brice, Krishnendu Chatterjee, K. S. Thejaswini}
\title{$\epsilon$-Stationary Nash Equilibria in Multi-player Stochastic Graph Games}
\keywords{Nash Equilibria, $\epsilon$-Nash equilibria, Approximation, Existential Theory of Reals}

\Copyright{Ali Asadi, L\'eonard Brice, Krishnendu Chatterjee, K. S. Thejaswini}

\author{Ali Asadi}{Institute of Science \& Technology Austria \and \url{https://ali-asadi.com/}}{ali.asadi@ista.ac.at}{https://orcid.org/0009-0005-2839-953X}{}

\author{Léonard Brice}{Université Libre de Bruxelles, Belgium \and \url{https://lnrdbrice.github.io/lnrdbrice/}}{leonard.brice@ulb.be}{https://orcid.org/0000-0001-7748-7716}{}

\author{Krishnendu Chatterjee}{Institute of Science \& Technology Austria \and \url{https://pub.ista.ac.at/~kchatterjee/}}{Krishnendu.Chatterjee@ist.ac.at}{https://orcid.org/0000-0002-4561-241X}{}

\author{K. S. Thejaswini}{Institute of Science \& Technology Austria \and \url{https://thejaswiniraghavan.github.io/}}{thejaswini.k.s@ista.ac.at}{https://orcid.org/0000-0001-6077-7514}{}

\category{} 

\relatedversion{This paper has been published at  IARCS Annual Conference on Foundations of Software Technology and Theoretical Computer Science}
\relatedversiondetails[linktext={https://doi.org/10.4230/LIPIcs.FSTTCS.2025.9}, cite=ABCT25]{Conference version}{ https://doi.org/10.4230/LIPIcs.FSTTCS.2025.9} 
\ccsdesc[500]{Theory of computation~Automata over infinite objects}

\funding{This work is a part of project VAMOS that has received funding from the European Research Council (ERC), grant agreement No 101020093.}

\EventEditors{C. Aiswarya, Ruta Mehta, and Subhajit Roy}
\EventNoEds{3}
\EventLongTitle{45th IARCS Annual Conference on Foundations of Software Technology and Theoretical Computer Science (FSTTCS 2025)}
\EventShortTitle{FSTTCS 2025}
\EventAcronym{FSTTCS}
\EventYear{2025}
\EventDate{December 17--19, 2025}
\EventLocation{BITS Pilani, K K Birla Goa Campus, India}
\EventLogo{}
\SeriesVolume{360}
\ArticleNo{34}
\begin{document}

\maketitle
\begin{abstract}
A strategy profile in a multi-player game is a Nash equilibrium if no player can unilaterally deviate to achieve a strictly better payoff. 
A profile is an $\epsilon$-Nash equilibrium if no player can gain more than $\epsilon$ by unilaterally deviating from their strategy.
In this work, 
we use $\epsilon$-Nash equilibria to approximate the computation of Nash equilibria.
Specifically, we focus on turn-based, multiplayer stochastic games played on graphs, where players are restricted to stationary strategies---strategies that use randomness but not memory.

The problem of deciding the constrained existence of stationary Nash equilibria---where each player's payoff must lie within a given interval---is known to be $\exists\mathbb{R}$-complete in such a setting (Hansen and Sølvsten, 2020). We extend this line of work to stationary $\epsilon$-Nash equilibria and present an algorithm that solves the following promise problem: given a game with a Nash equilibrium satisfying the constraints, compute an $\epsilon$-Nash equilibrium that \emph{$\epsilon$-satisfies} those same constraints---satisfies the constraints up to an $\epsilon$ additive error. Our algorithm runs in $\FNP^\NP$ time.

To achieve this, we first show that if a constrained Nash equilibrium exists, then one exists where the non-zero probabilities are at least an inverse of a double-exponential in the input. We further prove that such a strategy can be encoded using floating-point representations, as in the work of Frederiksen and Miltersen (2013), which finally gives us our $\FNP^\NP$ algorithm. 

We further show that the decision version of the promise problem is $\NP$-hard. 
Finally, we show a partial tightness result by proving a lower bound for such techniques: if a constrained Nash equilibrium exists, then there must be one where the probabilities in the strategies are double-exponentially small.

\end{abstract}
\section{Introduction}\label{sec:intro}

Modelling decentralised systems involving multiple agents requires understanding the interactions between different players, each with their own objectives. Stochastic games---graphs in which some nodes are controlled by agents and others are controlled by a random environment---provide a modelling framework for a wide range of domains, including epidemic processes~\cite{Lef81}, formal verification~\cite{FKNP11}, learning theory~\cite{AJKS21}, cyber-physical systems~\cite{SEC16}, distributed and probabilistic programs~\cite{dAHJ01}, and probabilistic planning~\cite{TKI10}.

These stochastic multi-player turn-based games are a specific class of models that capture interactions between multiple agents on such graph arenas. These games have been extensively studied, with several results characterizing their computational complexity~\cite{CMJ04,UW11} and others studying restrictions and different payoff functions that yield tractable fragments~\cite{UW11, BHT25}. While it is known that Nash equilibria (NEs) always exist in these settings, the equilibria that arise in  proofs are often adversarial in nature, offering poor outcomes for all players~\cite{CMJ04,Umm10}. A more compelling question is whether one can decide, and compute if possible, a Nash Equilibrium that satisfies certain payoff constraints. However, this problem becomes undecidable when players are allowed unbounded memory. Indeed, it is already known that even checking whether a player can achieve a payoff within a given interval is undecidable with just 10 agents~\cite[Theorem 4.9]{UW11}. Restricting players to deterministic (pure) strategies does not help, as the problem remains undecidable. Thus, for any hope of tractability, we must restrict each player to use bounded memory. In practice, this is also desirable: strategies requiring large memory lead to complex controllers and are harder to implement. Fortunately, finite-memory strategies can be encoded directly in the state space of the game, and in many cases, it suffices to consider stationary strategies, which require no memory at all.

 In this work, we only focus on Nash equilibria where each player is restricted to stationary strategies. Indeed, henceforth, whenever we mention strategies, we only mean stationary strategies, unless specified otherwise. However, it is known that the constrained existence problem even when restricted to such stationary strategies is $\exists\mathbb{R}$-complete~\cite{HS20}. 
Moreover, for the functional version of the problem, the probabilities used to indicate the distribution of these stationary strategies can have irrational values~\cite[Theorem 4.6]{UW11} even for games with just four players. Indeed, such strategies cannot be executed in real-world scenarios since it is hard to simulate choosing an edge with such irrational probabilities. 

The standard definition of Nash equilibrium requires that no player can improve their payoff by any non-zero amount by unilaterally deviating from their strategy even if the improvement is arbitrarily small~\cite{Nas50}. This strict notion of optimality contributes significantly to the computational complexity of deciding whether a constrained Nash equilibrium exists. From a practical perspective, however, it is reasonable to assume that players are unlikely to deviate from their current strategies for only negligible gains.

\subparagraph*{Our setting.} Motivated by this discussion, we consider the relaxed variant of the equilibrium concept that is also well-studied~\cite{LMM03, DGP09}. Specifically, we study the problem of computing $\epsilon$-approximate Nash equilibria---in which the expected utility of each player is required to be within $\epsilon$ 
of the optimum response to the other players’ strategies---while still satisfying given constraints on the players’ payoffs. Formally, we focus on multi-player turn-based games with terminal rewards and for a (stationary) strategy profile that is an $\epsilon$-NE that also approximately satisfies the constraints (the constraints are also satisfied up to an $\epsilon$-additive factor), assuming there is an exact stationary NE. 

The decision version of this problem is defined as follows: given such a game, determine whether either (i) there exists an exact NE that satisfies the constraints, or (ii) no $\epsilon$-approximate NE satisfies the constraints $\epsilon$-approximately. Instances that fall outside these two cases---that is, where an $\epsilon$-approximate equilibrium exists but no exact one does---are considered indeterminate, and the output may be arbitrary in those cases. Note that whenever there is an exact NE, by definition, there exists an $\epsilon$-NE. 

\subparagraph*{Our results.} We show that this approximate variant is computationally more tractable than the general problem of deciding the existence of a stationary constrained equilibrium. 
In particular, we present an algorithm that solves the problem within the class $\FNP$ with access to an $\NP$ oracle. 
We also show that the probabilities required for the exact version of the problem can be double-exponentially small, proving some evidence of hardness of the problem. 
Finally, we show that the decision version of the problem is $\NP$-hard.

\paragraph*{Technical overview} 
\subparagraph*{Upper bound.}  Recall our result that the problem lies in $\FNP^\NP$. Our approach relies on showing that if a stationary equilibrium exists, then there is one in which the probability values used in the strategies are not too small---specifically, they are lower-bounded by the inverse of a doubly exponential function of the input size. 
This structural property is established using techniques from the work of Hansen, Koucký, and Miltersen on zero-sum concurrent reachability games~\cite{HKM09}. However, this bound alone is insufficient, because representing such tiny probabilities explicitly requires at least exponentially many bits to represent using fixed-point presentation.

To represent the probabilities used in strategies more succinctly than with fixed-point representation, we use floating-point values, a standard tool in numerical analysis. While floating-point representations cannot express the exact probabilities required in an exact Nash equilibrium, they are sufficient for representing approximate equilibria. We prove (see \cref{lem:floating-point-eps-nash}) that if an exact Nash equilibrium exists, then there is also an $\epsilon$-NE that can be encoded using floating-point numbers.

The proof relies on two ideas from the work of Frederickson and Milterson~\cite{FM13}. First, any probability distribution can be approximated by one using floating-point values, such that the difference between them is small in a precise sense. Second, when all the probabilities in a Markov chain are approximated in this way, the expected value of the chain changes by at most a polynomial function of the approximation error. To formalise this notion of approximation, a distance metric introduced by Solan~\cite{Sol03} is used to show how values of Markov chains behave under small perturbations. 

Using these results, we design an algorithm for the problem. The algorithm first guesses a stationary approximate equilibrium that satisfies the payoff constraints up to the given error. Our earlier results ensure that such a guess can be represented using floating-point numbers with only polynomially many bits. The algorithm then verifies that this strategy profile is indeed an approximate equilibrium by checking, for each player, whether any alternative strategy would improve their payoff by more than the allowed margin of error. This check reduces to solving a Markov decision process for each player, where the strategies of the other players are fixed. We show that solving this single-player game can be done with a single call to an $\NP$ oracle for each player (see \cref{Result:Approximate_MDPs}). Therefore, the search problem lies in the class $\FNP^\NP$, and the corresponding decision problem is in $\NP^\NP$.

\subparagraph*{Small probabilities.} We first consider the limits of the above approach. Prior work by Deligkas, Fearnley, Melissourgos, and Spirakis~\cite{DFMS18} observed that the class $\epsilon$-$\exists\mathbb{R}$, which captures the complexity of deciding whether approximate solutions exist to systems of real-valued constraints, is polynomial time reducible to the  class $\exists\mathbb{R}$ itself. This equivalence in the complexity of $\exists\Rb$ and its approximation version hinges on the fact that solutions to such existential theory of the reals (ETR) problems can involve values that are double-exponentially large. To address this, they identified a restricted fragment of ETR for which $\epsilon$-approximate solutions can be computed using a quasi-polynomial time approximation scheme (QPTAS).

In our setting, we use a key structural result (see \cref{lem:double_exp_NE_exists}) showing that stationary Nash equilibria, when they exist, can involve probabilities that are double-exponentially small, but not smaller. One might suspect that such small probabilities are merely theoretical and not actually be needed to represent probabilities occurring in such equilibria. However, in Section~\ref{sec:smallprobs}, we present a concrete example with just five players where double-exponentially small probabilities are indeed necessary for an exact equilibrium to exist. We note that previous results establishing $\exists\mathbb{R}$-completeness for related problems required seven players, underscoring the tightness of our example. 

\subparagraph*{Hardness results.} Finally, in \cref{sec:lowerbounds}, we also show that even the decision version of the promise problem is $\NP$-hard, via a reduction from the classical $\THREESAT$ problem. While it was already known that the constrained existence of stationary Nash equilibria is NP-hard~\cite[Theorem 4.4]{UW11} and later shown to be even $\exists\Rb$-complete, we prove that even the relaxed version of the problem---deciding whether an approximate (i.e., $\epsilon$-close) stationary Nash equilibrium exists that satisfies the given constraints---is also $\NP$-hard.
\subparagraph*{Related work.}
The approximation techniques we use to establish our $\FNP^\NP$ upper bound have appeared in various forms in the literature, particularly in the context of approximating values in concurrent games. For example, Frederiksen and Miltersen~\cite{FM13} introduced similar floating-point approximation methods to compute values of concurrent reachability games up to arbitrary precision. Their work placed the problem in the complexity class $\TFNP^\NP$. However, for this class of games, no corresponding hardness results—such as the $\NP$-hardness we establish for our setting—were previously known.

These techniques have also been extended to other objectives in the same setting of concurrent games. In particular, Asadi, Chatterjee, Saona, and Svoboda~\cite{ACSS24} applied related methods to show $\TFNP^\NP$ upper bounds for games with stateful-discounted objectives and parity objectives.

More recently, Bose, Ibsen-Jensen, and Totzke~\cite{BIT24} studied $\epsilon$-approximate Nash equilibria in concurrent and partial-observation settings, also using approximation-based techniques. However, their work differs from ours in two key ways. First, they do not address the \emph{constrained existence problem} for $\epsilon$-Nash equilibria. We note that the unconstrained variant of the problem, which they consider, is significantly easier in terms of computational complexity. Second, they assume the number of players is fixed, whereas in our setting, the number of players is part of the input. To our knowledge, our work is the first to analyse constrained $\epsilon$-Nash equilibria in multi-player games where the number of players is not fixed.

\section{Preliminaries}
We assume that the reader is familiar with the basics of probability and graph theory. However, we define some concepts for establishing notation.

\subp{Probabilities} Given a (finite or infinite) set of outcomes $\Omega$ and a probability measure $\prob$ over $\Omega$, let $X$ be a random variable over $\Omega$, i.e., a mapping $X: \Omega \to \Rb$. We write $\Eb^\prob[X]$, or simply $\Eb[X]$, for the expected value of $X$, when defined.
Given a finite set $S$, a \emph{probability distribution} over $S$ is a mapping $d: S \to [0,1]$ that satisfies the equality $\sum_{x \in S} d(x) = 1$.

\subp{Size of numbers}
Integers are always implicitly assumed to be represented with their binary encoding.
For a given integer $n \in \Nb$, we therefore write $\bit(n) = \lceil \log_2(n+1) \rceil$ for the number of bits required to write $n$.
Similarly, unless stated otherwise, a rational number $\frac{p}{q}$, where $p$ and $q$ are co-prime, is represented by the pair of the encodings of $p$ and $q$, and we write $\bit\left( \frac{p}{q} \right) = \bit(p) + \bit(q) + 1$.

\subp{Graphs and games}
A directed graph $(V,E)$ consists of a set $V$ of \emph{vertices} and a set $E$ of ordered pairs of vertices, called \emph{edges}. 
For simplicity, we often write $uv$ for an edge $(u, v)\in E$.
A \emph{path} in the directed graph $(V, E)$ is a (finite or infinite) ordered sequence of vertices from $V$ such that every pair of two consecutive elements is an edge. 
A \emph{cycle} is a path with distinct vertices such that the last element of the sequence and the first one also forms an edge.

Throughout this paper, we use the word \emph{game} for multiplayer simple quantitative turn-based games played on graphs.

\begin{definition}[Game]
    A \emph{game} is a tuple 
    $\Game$ 
    that consists of:
    \begin{itemize}
        \item a directed graph $(V, E)$, called the \emph{underlying} graph of $\Game$;

        \item a finite set $\Pi$ of \emph{players};

        \item a partition $(V_i)_{i \in \Pi \cup \{?\}}$ of the set $V$, where $V_i$ denotes the set of vertices \emph{controlled} by player $i$, and the vertices in $V_?$ are called \emph{stochastic vertices};

        \item a \emph{probability function} $\p: E(V_?) \to [0, 1]$, such that for each stochastic vertex $s$, the restriction of $\p$ to $E(s)$ is a probability distribution;

        \item a mapping $\mu: T \to [0,1]^\Pi$ called \emph{payoff function}, where $T$ is the set of \emph{terminal vertices}, that is, vertices of the graph $(V, E)$ that have no outgoing edges.
        We also write $\mu_i$, for each player $i$, for the function that maps a terminal vertex $t$ to the $i^\text{th}$ coordinate of the tuple $\mu(t)$.

        \item a vertex $v_0\in V$, which is the \emph{initial} vertex of the game.
    \end{itemize}
\end{definition}
When referring to a game $\Game$, we will often use the notations $V$ for vertices, $E$ for edges, $\Pi$ for players, and so on without necessarily recalling them.

\begin{definition}[Markov decision process, Markov chain]
    A \emph{Markov decision process} is a game with one player.
    A \emph{Markov chain} is a game with zero players.
\end{definition}

\subp{Plays}
We call \emph{play} a path in the underlying graph that is infinite, or whose last vertex is a terminal vertex. 
The payoff functions $\mu$ and $\mu_i$, defined only for terminal vertices, are extended naturally to  plays as follows: for a play of the form $ht$, with $t \in T$, we define $\mu(ht) = \mu(t)$, and for an infinite play $\pi$, we define $\mu(\pi) = (0)_{i \in \Pi}$ (all players receive the payoff $0$ in an infinite path where no terminal vertex is reached).

\subp{Strategies, and strategy profiles}
In this paper, by \emph{strategy}, we mean a stationary strategy, that is, a strategy that only depends on the current vertex.

    Thus, in a game $\Game$, a \emph{strategy} for player $i$ is a mapping $\sigma_i$ that maps each vertex $v \in V_i$ to a probability distribution over the neighbours of $v$.
A \emph{strategy profile} for a subset $P \subseteq \Pi$ is a tuple $(\sigma_i)_{i \in P}$, that we usually write $\bsigma$ if $P = \Pi$, and $\bsigma_{-i}$ if $P = \Pi\setminus\{i\}$ for some player $i$. 
We then write $(\sigma_{-i}, \sigma^\prime_i)$ to denote the strategy profile $\btau$ defined by $\tau_i = \sigma^\prime_i$ and $\tau_j = \sigma_j$ for $j \neq i$. 

A strategy profile $\bsigma$ for all players in the game $\Game$ defines a probability measure $\prob_\bsigma$ over plays---which turns the payoff functions $\mu_i$ into random variables.
We then write $\Eb(\bsigma)$ for the expectation operator $\Eb^{\prob_\bsigma}$.
If the game is a Markov chain, then we just write $\Eb$ to represent this value. We say the maximum payoff of an MDP is the maximum of all expectations over all stationary strategies and write $\val = \max_\sigma \Eb(\bsigma)$ to denote it.

\subp{Equilibria}
In such games, we study \emph{equilibria}, that is, strategy profiles that offer some stability guarantees.
The most famous notion of equilibrium is the \emph{Nash equilibrium}~\cite{Nas50}.

\begin{definition}[Nash equilibrium]
    A \emph{Nash equilibrium}, or \emph{NE} for short, in the game $\Game$ is a strategy profile $\bsigma$ such that for each player $i$ and every strategy $\sigma'_i$, we have $$\Eb(\bsigma_{-i}, \sigma'_i)[\mu_i] \leq \Eb(\bsigma)[\mu_i].$$
\end{definition}

A classical quantitative relaxation of Nash equilibria is that of $\epsilon$-approximate Nash equilibria, or just $\epsilon$-Nash equilibria. 

\begin{definition}[$\epsilon$-Nash equilibrium]
    Let $\epsilon > 0$.
    An \emph{$\epsilon$-Nash equilibrium}, or \emph{$\epsilon$-NE} for short, in the game $\Game$ is a strategy profile $\bsigma$ such that for each player $i$ and every strategy $\sigma'_i$, we have $\Eb(\bsigma_{-i}, \sigma'_i)[\mu_i] - \epsilon \leq \Eb(\bsigma)[\mu_i]$.
\end{definition}

\subp{Problem}
A classical decision problem about Nash equilibria is the following one.

\begin{problem}[Constrained existence problem of Nash equilibria]\label{prob:CENE}
    Given a game $\Game$ and two vectors $\bx, \by \in \Qb^\Pi$, called \emph{threshold vectors}, does there exist an NE $\bsigma$ in $\Game$ that satisfies the inequality $\bx \leq \Eb(\bsigma)[\mu] \leq \by$?
\end{problem}

\cref{prob:CENE} is known to be undecidable if the players are allowed to use memory~\cite{UW11}, which is excluded by our formalism here.
Where players are restricted to \emph{stationary} strategy---as in our definition---this problem is known to be $\exists\Rb$-complete~\cite{HS20}, where $\exists\Rb$ is the complexity class of problems that can be reduced to the satisfiability of a formula in the existential theory of the reals.
In order to obtain a more tractable problem, we consider here an approximated version defined using the notion of $\epsilon$-NEs: we wish to obtain, for sure, the answer \emph{yes} when an NE satisfying the constraint exists, and to obtain, for sure, the answer \emph{no} when we are far from having such an NE because there is not even an \emph{$\epsilon$-NE} that satisfies the constraint \emph{up to $\epsilon$}---but on the limit case, we can accept any answer.
Formally, we want an algorithm that solves the following problem.

\begin{problem}[Approximate constrained existence problem of Nash equilibria]
    Given a game $\Game$, two vectors $\bx, \by \in \Qb^\Pi$, and a rational number $\epsilon > 0$ input using fixed point representation, such that either:
    \begin{enumerate}
        \item\label{itm:positive_instance} there is an NE $\bsigma$ in $\Game$ with $x_i \leq \Eb(\bsigma)\left[ \mu_i \right] \leq y_i$ for each player $i$; or 

        \item there is no $\epsilon$-NE in $\Game$ satisfying $x_i - \epsilon \leq \Eb(\bsigma)\left[ \mu_i \right] \leq y_i + \epsilon$ for each player $i$;        
    \end{enumerate}
    holds, then are we in case~\ref{itm:positive_instance}?
\end{problem}

We call \emph{functional} approximate constrained NE problem the functional version of this problem, in which the answer \emph{yes} is replaced by a succinct representation of an $\epsilon$-NE that satisfies the constraint up to $\epsilon$.
In the sequel, we show a $\FNP^\NP$ upper bound on this functional problem, which implies an $\NP^\NP$ upper bound on the non-functional one; and an $\NP$ lower bound on the non-functional approximated problem, which implies an $\FNP$ lower bound on the functional one.

\section{Finding $\epsilon$-Nash equilibria}\label{sec:UB}

    In this section, we present the $\FNP^\NP$ upper bound for the functional approximate constrained NE problem.
    We always use $n$ to represent the number of vertices in the game and $m$ to represent the number of edges in the game,  and $\tau$ which corresponds to the  bit size used to represent the probabilities. 

    Throughout this section, we assume that we are given an instance $(\Game, \bx, \by, \epsilon)$ of the functional approximate constrained problem.
We call \emph{($\epsilon$-)constrained NE} an ($\epsilon$-)NE in $\Game$ such that each player $i$'s payoff lies in the interval $[x_i, y_i]$ (resp. $[x_i-\epsilon, y_i+\epsilon]$).
    The rest of this section is dedicated to the proof of the following theorem.


\begin{restatable}{theorem}{NPcoNPUB}
\label{thm:np-np-upper-bound}
    There exists an $\FNP^\NP$ procedure to solve the functional approximated constrained problem of NEs.
\end{restatable}
We begin by defining the existential theory of the reals  and showing how the problem of computing stationary Nash equilibria can be encoded using sentences in it.

By applying a result of Basu, Pollack, and Roy~\cite{BPR96} (restated in \cref{lem:radius-of-etr}), it can be shown that any solution to the resulting system of polynomial inequalities lies within a ball of at most double-exponential radius. This implies that, in the worst case, the smallest non-zero probabilities in a stationary equilibrium may be as small as inverse double-exponential in the size of the input. 

To represent such small probabilities succinctly, we introduce a floating-point representation that allows encoding of some double-exponentially small values using only polynomially many bits. 
Next, using a result of Solan~\cite{Sol03} (see \cref{Result: Continuity of value in MC with reachability}), we observe that small perturbations in transition probabilities lead to only small changes in the value of the associated Markov chain. This continuity result ensures that approximate equilibria obtained where probabilities are slightly perturbed still yield approximately correct values.
This shows that it is enough to represent the probabilities using floating-point numbers with polynomial precision.

Finally, we provide a $\coNP$ procedure (see \cref{Result: Approximate reachability value in MDPs}) to decide whether the maximum payoff in a Markov Decision Process (MDP), where probabilities are given in floating-point representation, lies above or below a given threshold up to an additive $\epsilon$.

Together, these components give us our desired $\FNP^\NP$ procedure: we guess a stationary strategy profile that approximately satisfies the constraints and verify it using the above results, and use the fact that such strategies admit succinct polynomial-size representations.






\subp{Existential theory of the reals} 
A sentence in the \emph{existential theory of the reals} is of the form
\[
\phi \;=\; \exists x_1\, \exists x_2\;\cdots\;\exists x_k\;F(x_1,\dots,x_k)\,,
\]
where $F(x_1,\dots,x_k)$ is a \emph{quantifier-free} formula in the language of ordered fields over the reals $\Rb$. More concretely, atomic sub-formulas are
\[
    p(x_1,\dots,x_k)=0
    \quad\text{or}\quad
    p(x_1,\dots,x_k)<0,
\]
where $p\in \Rb[x_1,\dots,x_k]$ is a polynomial over $\Rb$, and arbitrary formulas are built by Boolean connectives.

\begin{lemma}
\label{lem:double-exponential-patience}
    If there exists a constrained
    Nash equilibrium $\bsigma$, then there exists a constrained
    Nash equilibrium such that the non-zero probabilities are at least $2^{-\tau 2^{O(m|\Pi|)}}$.
\end{lemma}
\begin{proof}
    This statement is similar to a result in the work of Hansen, Koucký, and  Miltersen~\cite[Thm. 4]{HKM09}. We write the ETR sentence for solving constrained stationary Nash equilibrium in the following proposition. 
    \begin{restatable}{proposition}{ETRSentence}
        There exists an ETR sentence for computing constrained stationary Nash equilibrium with $O(m|\Pi|)$ variables, where all polynomials are of degree at most 2.
    \end{restatable}
        \begin{claimproof}
    The following part is similar to the ETR sentence for constrained stationary Nash equilibria found in the work of Ummels and Wojtczak~\cite[Theorem 4.5]{UW11}. However, we require the explicit bounds on the number of variables, the degree polynomials, and the length of the sentence. Thus, we rewrite the explicit ETR sentence. 
    First, by the assumption, there exists a strategy profile that is a stationary Nash equilibrium $\bsigma$. Let $S$ be the set of all non-zero edges in the strategy profile $\bsigma$. Furthermore, we define the set of vertices $V_S$ from which terminals can be reached with non-zero probability.
    Second, we have a formula that states that the variables $p_{vw}$ indeed describe a strategy. We enforce that the non-zero variables belong to the set $S$.  We further ensure that for stochastic vertices, the variable $p_{vw}$ encodes exactly the value dictated by the probability function $\p$ by the stochastic vertex. 
    This involves satisfying the following constraints:
    \begin{enumerate}
        \item $p_{vw}> 0$; for each $vw\in S$
        \item $p_{vw} \leq 1$; for each $vw\in S$
        \item $p_{vw}= 0$; for each $vw\notin S$
        \item $\sum_{w\in E(v)}p_{vw}=1$; for each non-stochastic vertex $v$
        \item $p_{vw} =  \p(vw)$; for all stochastic vertices $v$.
    \end{enumerate}
    The variables $r_v^i$ correspond to the payoff player $i$ receives by starting from the vertex $v$. The payoffs $r_v^i$ are equal to $0$ for all vertices in $V_S$ because no terminal is reachable from such vertices. 
    The payoffs for other vertices satisfy the following constraints.    
    \begin{enumerate}
  \setcounter{enumi}{5}
        \item $r^i_t =\mu_i(t)$, for $t \in T$;
        \item $ r^i_v = 0$, for $v\notin V_S$;
        \item $r^i_v = \sum_{w\in E(v)}p_{vw}\,r^i_v$, for $v\in V_S\setminus T$;
    \end{enumerate}
    We now check that there is no unilateral deviation for the players by adding the following constraints. This following check is enough because, when fixing the strategies of all players except player $i$, the resulting structure is an MDP. For MDPs, Bellman equations characterise the value of each state and the equation below enforces that the value at a state $v$ is at least as large as the maximum value of any successor, thereby ensuring consistency with the Bellman equations of the induced MDP~\cite[Theorem 10.109]{BK08}. 
    \begin{enumerate}
        \setcounter{enumi}{8}
        \item $ r^i_v \ge r^i_w$, for $v \in V_i, w \in E(v)$;
    \end{enumerate}
    Finally, we add the auxiliary variables $g_{vw}$ to keep track of the size of $p_{vw}$, which satisfy the following constraint:
    \begin{enumerate}
    \setcounter{enumi}{9}
        \item $g_{vw}\,p_{vw} = 1$, for each $vw \in S$;
    \end{enumerate}
    Thus, our ETR statement is: 
    $\exists p\:\exists r\:\exists g$ that satisfies constraints $1$-$10$, completing the proof. 
    \end{claimproof}
    
    The following result from the work of Basu, Pollack, and Roy~\cite{BPR96} guarantees that if a system of polynomial inequalities is satisfiable, then there is a valuation satisfiable within a ball with radius double-exponentially large with respect to the input size. Before we state the result, we recall that a partition of semi-algebraic subsets of $\mathbb{R}^k$ generated by a set of polynomials $\Pc$ is the collection of all nonempty sets of the form
        $\{x\in\Rb^k \mid \textrm{sign}(p(x))=\sigma(p)\text{ for all }p\in \Pc\}$,
        where  $\sigma:\Pc\to\{-1,0,1\}$ ranges over all possible sign assignments. Intuitively, the partition is obtained by slicing $\Rb^k$ according to which side of each polynomial in $\Pc$ a point lies on (positive, zero, or negative).

    \begin{lemma}[\protect{\cite[Theorem 1.3.5]{BPR96}}]
    \label{lem:radius-of-etr}
        Given a set $\Pc$ of polynomials of degree $d$ in $k$ variables with coefficients in $\Zb$ of bit-sizes $\tau$, then there exists a ball of radius $2^{\tau d^{O(k)}}$ that intersects every part of the partition of semi-algebraic subsets of $\Rb^k$ generated by $\Pc$ (cells of $\Pc$).
    \end{lemma}
    Note that each of the constraints in the ETR sentence is a polynomial of degree at most 2. The total number of variables is $O(m|\Pi|)$. 
    Note that we had introduced auxiliary variables $g_{vw}$ and required that $g_{vw}\,p_{vw} = 1$, for each $vw \in S$. Since this is also a polynomial in $\Pc$, using \Cref{lem:radius-of-etr}, for all $vw \in S$, we have $g_{vw} \le 2^{\tau 2^{O(m|\Pi|)}}$. Consequently, we have $p_{vw} \ge 2^{-\tau 2^{O(m|\Pi|)}}$, which completes the proof. 
\end{proof}

We recall some results used in the work of Frederiksen and Milterson~\cite{FM13} where they adapted arguments that are standard in the context of numerical analysis to provide better algorithms for approximating the values of concurrent reachability games.

\subp{Floating-point number representation} We define the set of floating-point numbers with precision $\ell$ as 
\begin{align*}
  \Fc(\ell) 
    \defas \left\{ m \cdot 2^{e} \quad \mid \quad m \in \{0, \cdots, 2^\ell-1\}, \quad e \in \Zb \right \} \,.
\end{align*}
The floating-point representation of an element $x = m \cdot 2^e \in \Fc(\ell)$ uses $\bit(m) + \bit(|e|)$ bits. We define the relative distance of two positive real numbers $x, \widetilde{x}$ as 
\[
  \rel(x, \widetilde{x}) 
    \defas \max \left\{ \frac{x}{\widetilde{x}}\,, \frac{\widetilde{x}}{x} \right\} - 1 \,.
\]
Intuitively, $\rel(x,\widetilde{x})$ describes the multiplicative distance. Observe that if the values are closer to $0$, but $\epsilon$ apart for some small $\epsilon$, their relative distance is larger than if these values are much larger and only $\epsilon$ apart.  
We say $x$ is $(\ell, i)$-close to $\widetilde{x}$ if 
    $\rel(x, \widetilde{x}) \le (1 - 2^{1 - \ell})^{-i} - 1$,
where $\ell$ is a positive integer and $i$ is a non-negative integer.
$(\ell,i)$-closeness measures how far two numbers can differ after $i$ steps of rounding in an 
$\ell$-bit floating-point system.
Here, $\ell$ controls precision (smaller gaps between representable numbers), and $i$ allows for cumulative rounding error.
So $x$ is $(\ell,i)$-close to $\widetilde{x}$ if their relative difference is within about 
$i\cdot 2^{1-\ell}$, that is, within $i$ units of machine precision.
\subp{Arithmetic operations} We define $\oplus^{\ell}, \ominus^{\ell}, \otimes^{\ell}, \oslash^{\ell}$ as finite precision arithmetic operations $+, -, *, /$ respectively by truncating the result of the exact arithmetic operation to $\ell$ bits. We drop the superscript $\ell$ if context is clear.

\subp{Floating-point probability distribution representation}
We denote by $\Dc(\ell)$ the set of all floating-point probability distributions with precision $\ell$. 
A probability distribution $\distribution \in \Delta([t])$ belongs to $\Dc(\ell)$ if there exists $w_1, w_2, \cdots, w_t \in \Fc(\ell)$ such that
\begin{itemize}
    \item For all $i \in [t]$, we have $\distribution(i) = \frac{w_i}{\sum_{j \in [t]} w_j}$; and
    \item $\sum_{j \in [t]} w_j$ and 1 are $(\ell, t)$-close.
\end{itemize}
We define the relative distance $\rel$ for probability distributions as
$
    \rel (\distribution, \widetilde{\distribution}) \defas \max \{ \rel(\distribution(i), \widetilde{\distribution}(i)) : i \in [t]\}
$. 
We say $\distribution$ is $(\ell, i)$-close to $\widetilde{\distribution}$ if 
    $\rel(\distribution, \widetilde{\distribution}) \le (1 - 2^{1 - \ell})^{-i} - 1$,
where $\ell$ is a positive integer and $i$ is a non-negative integer.

\begin{lemma}
\label{lem:floating-point-eps-nash}
    If there exists a constrained stationary Nash equilibrium $\bsigma$, then there exists an $(32n^22^{-\ell})$-constrained stationary Nash equilibrium $\bsigma'$ such that for all players $i$ and vertices $v \in V_i$, we have $\bsigma'(v) \in \Dc(\ell)$, where $\ell \ge 1000n^2$.
\end{lemma}
\begin{proof}
    We first recall the result in the work of Frederiksen and Milterson~\cite{FM13} related to approximation of probability distributions by floating-point distributions. 
    \begin{lemma}[\protect{\cite[Lemma 5]{FM13}}]
    \label{Result: Approximation of prob distribution}
        Consider $x_1, \cdots, x_t \in \Fc(\ell)$. 
        Let $\distribution(i) \defas x_i \oslash \left ( \bigoplus_{j=1}^t x_j \right )$. 
        Then, there exists $\widetilde{\distribution} \in \Dc(\ell)$ such that for all $i$, we have $\widetilde{\distribution}(i) = \distribution(i) / 
     \left (\sum_{j=1}^{t} \distribution(j) \right )$, and $\distribution$ and $\widetilde{\distribution}$ are $(\ell, 2t)$-close. 
    \end{lemma}
     A stationary strategy comprises of probability distributions over actions for each state. We can truncate a strategy by truncating each of these probability distributions. Let the strategy profile $\bsigma^\prime$ be the truncation of the strategy profile $\bsigma$ defined in the above result. Therefore, for all players $i$ and vertices $v \in V_i$, we have that $\bsigma(v)$ and $\bsigma^\prime(v)$ are $(\ell, 2n)$-close. Consequently, for all vertices $v, w \in V$,
    \begin{align}
        \rel(\bsigma(v)(w), \bsigma^\prime(v)(w))
        &\le \frac{1}{(1 - 2^{1 - \ell})^{2n}} - 1 \nonumber \\
        &\le \frac{1}{1 - (2n)2^{1 - \ell}} - 1 & (\text{Bernoulli inequality}) \nonumber \\
        &\le \frac{(2n)2^{1 - \ell}}{1 - (2n)2^{1 - \ell}} & (\text{rearrange}) \nonumber \\
        &\le 4n2^{-\ell} & \left (\ell \ge 1000 n^2 \right )
        \label{eq:rel-bound}
    \end{align}
    We then recall the result in the work of Solan~\cite{Sol03}. This result provides an upper bound on the difference between the reachability values of two Markov chains based on the relative distance of their transition functions.
    \begin{lemma}[\protect{\cite[Thm. 6]{Sol03}}]
    \label{Result: Continuity of value in MC with reachability}
        Consider two MCs $\MChain$ and $\widetilde{\MChain}$ with identical vertex sets and a target set $T$. Let $\mu$ be a reward function, where $\mu(v) = 1$ for all vertices $v \in T$ and $\mu(v) = 0$ for all vertices $v \in V \setminus T$.
        We denote by $\val$ and $\widetilde{\val}$ the expected payoff of $\MChain$ and $\widetilde{\MChain}$ respectively. 
        Fix $\epsilon \defas \max_{v, w \in V} \rel(\p(vw), \widetilde{\p}(vw))$. Then, we have
        \[
            |\val - \widetilde{\val}| \le 4 n \epsilon \,.
        \]
    \end{lemma}
    By \Cref{eq:rel-bound} and \Cref{Result: Continuity of value in MC with reachability}, we bound the the difference of payoffs for $\bsigma$ and $\bsigma^\prime$. For all players $i$, we have $|\mathbb{E}(\bsigma)[\mu_i] - \mathbb{E}(\bsigma^\prime)[\mu_i]| \le 16n^22^{-\ell}$. Since the strategy profile $\bsigma$ is a constrained Nash equilibrium, $\bsigma^\prime$ is a $(32n^22^{-\ell})$-constrained Nash equilibrium, which completes the proof.
\end{proof}

In the above result, we show that it is sufficient to consider floating-point representable strategies to solve the problem. We now present the $\NP^\NP$ procedure. We first recall the result in the work of Frederiksen and Milterson~\cite{FM13} related to computing the approximate value of MCs with floating-point distributions.

\begin{lemma}[\protect{\cite[Thm. 4]{FM13}}]
\label{Result: Approximate reachability value}\label{}
    Consider an MC $\MChain$ and a target set $T$. Let $\mu$ be a reward function, where $\mu(v) = 1$ for all vertices $v \in T$ and $\mu(v) = 0$ for all vertices $v \in V \setminus T$.
    For all vertices $v \in V$, we have $\p(v, .) \in \Dc(\ell)$ where $\ell \ge 1000 n^2$. 
    Then, there exists a polynomial-time algorithm that for all vertices $v \in V$, computes an approximation $r \in \Fc(\ell)$ for $\mathbb{E}[\mu]$ such that 
    \[
        |r - \mathbb{E}[\mu]| \le 80 n^4 2^{-\ell}\,,
    \]
    where $\val_T(v)$ is the reachability value for the vertex $v$.
\end{lemma}

\begin{lemma}
\label{Result: Approximate reachability value in MDPs}\label{Result:Approximate_MDPs}
    The problem of deciding if the payoff for MDPs is below a threshold up to an additive error is in $\coNP$ where the input is a MDP $\MDProc$, a reward function $\mu$, a vertex $v$, a threshold $0 \le \alpha \le 1$, an additive error $\epsilon = 2^{-\kappa}$ and a positive integer $\ell$ such that, for all stochastic vertices $w \in V$, we have 
    \[
        \p(w, .) \in D(\ell), \quad \ell \ge 1000 n^2 + \kappa \,.
    \]
    Define $\val \defas \sup_{\sigma} \mathbb{E}(\sigma)[\mu]$. Note that the numbers $\alpha$ and $\epsilon$ are represented in fixed-point binary and the NP procedure is such that 
    \begin{itemize}
        \item If $\alpha \le \val - \epsilon$, then it outputs \emph{no}; and
        \item If $\alpha \ge \val + \epsilon$, then it outputs \emph{yes}.
    \end{itemize}
\end{lemma}

\begin{proof}
    We first present the $\coNP$ procedure and then prove its soundness and completeness.

    \smallskip\noindent{\em Procedure.} 
    The procedure guesses a pure stationary strategy $\sigma$ for the player. 
    Note that the size of the representation of a pure stationary strategy is polynomial with respect to the size of the representation of $\MDProc$. 
    By fixing $\sigma$, we obtain an MC $\MChain$. 
    We denote by $\val_\sigma$ the value $\mathbb{E}(\sigma)[\mu]$. 
    By \Cref{Result: Approximate reachability value}, there exists a polynomial time algorithm that computes an $\epsilon$-approximation $\widehat{\val}_\sigma$ of $\val_\sigma$. 
    Our procedure outputs \emph{no} if $\alpha \le \widehat{\val}_\sigma$. 
    If there exists no such pure stationary strategy, the procedure outputs \emph{yes}.
    
    \smallskip\noindent{\em Completeness.} 
    If $\alpha \le \val - \epsilon$, then, by \cite{FV97}, there exists a pure stationary strategy $\sigma$ such that $\val = \val_\sigma$. 
    The procedure non-deterministically guesses $\sigma$. 
    By \Cref{Result: Approximate reachability value}, we have $\widehat{\val}_\sigma + \epsilon \ge \val_\sigma$. 
    Therefore, we have $\alpha \le \widehat{\val}_\sigma$, and the procedure outputs \emph{no}.

    \smallskip\noindent{\em Soundness.} 
    If $\alpha \ge \val + \epsilon$, then for all pure stationary strategies $\sigma$, we have $\alpha \ge \val_\sigma + \epsilon$. 
    By \Cref{Result: Approximate reachability value}, we have $\widehat{\val}_\sigma - \epsilon \le \val_\sigma$. 
    Therefore, $\alpha \ge \widehat{\val}_\sigma$ which implies that the procedure outputs \emph{yes} and yields the result.
\end{proof}

\begin{proof}[Proof of \Cref{thm:np-np-upper-bound}]
    We first present the procedure and then prove its soundness and completeness.

    \smallskip\noindent{\em Procedure.} Let $\ell$ be sufficiently large. The procedure guesses a  strategy profile and verifies that it is an $\epsilon/8$-constrained stationary Nash equilibrium $\bsigma$. 
    By \Cref{lem:double-exponential-patience} and \Cref{lem:floating-point-eps-nash}, the size of the representation of $\bsigma$ is polynomial with respect to the size of game and $\bit(\epsilon)$. By fixing the strategy profile $\bsigma$, we obtain an MC $\MChain$. We denote by $\val_i$ the payoff of $\MChain$ where the reward function is $\mu_i$. By \Cref{Result: Continuity of value in MC with reachability}, the procedure computes the $\epsilon/8$-approximate payoff $\widehat{\val}_i$ in polynomial time. For each player $i$, by fixing the strategy profile $\bsigma_{-i}$, we obtain an MDP $\MDProc_i$. We denote by $\widetilde{val}_i$ the value of $\MDProc_i$ where the reward function is $\mu_i$. Then, the procedure checks if there exists an $\epsilon$-unilateral deviation for player $i$ by deciding if $\widetilde{\val}_i$ for is at most $\widehat{\val}_i(v) + 3/4\epsilon$ up to additive error $\epsilon/8$. The procedure finally checks if the payoff constraints are satisfied.
    
    \smallskip\noindent{\em Soundness.} 
    We assume that $\bsigma$ is not an $\epsilon$-constrained Nash equilibrium, i.e., there exists a player $i$ such that $\epsilon$-unilateral deviation is possible. Equivalently, we have $\val_i \le \widetilde{\val}_i - \epsilon$. Therefore, we have 
    \begin{align*}
        \widehat{\val}_i + 3/4\epsilon &\le \val_i + 7/8\epsilon & \left(\widehat{\val}_i \text{ is an $\epsilon/8$-approximation of } \val_i \right )\\
        &\le \widetilde{\val}_i - \epsilon/8 \,. &(\text{$\epsilon$-unilateral deviation})
    \end{align*}
    However, for this case, the $\coNP$ procedure defined in \Cref{Result: Approximate reachability value in MDPs}, successfully outputs \emph{no}. Therefore, our procedure does not output $\bsigma$ as $\epsilon$-constrained stationary Nash equilibrium, which yields the soundness of our procedure.

    \smallskip\noindent{\em Completeness.} By \Cref{lem:double-exponential-patience} and \Cref{lem:floating-point-eps-nash}, there exists an $\epsilon/8$-constrained stationary Nash equilibrium $\bsigma$ which is polynomial-size representable. 
    The procedure non-deterministically guesses $\bsigma$. 
    Since the strategy profile is $\epsilon/8$-Nash equilibrium, for all players $i$, we have $\widetilde{\val}_i \le \val_i + \epsilon/8$. Therefore, we get
    \begin{align*}
        \widehat{\val}_i + 3/4\epsilon &\ge \val_i + 5/8\epsilon & \left(\widehat{\val}_i \text{ is an $\epsilon/8$-approximation of } \val_i \right )\\
        &\ge \widetilde{\val}_i + 4/8\epsilon &(\widetilde{\val}_i \le \val_i + \epsilon/8)\\
        &\ge \widetilde{\val}_i + \epsilon/8 \,.
    \end{align*}
    Thus, the $\coNP$ procedure outputs \emph{yes}. Therefore, our procedure successfully decides that $\bsigma$ is a $\epsilon$-constrained stationary Nash equilibrium. This yields the completeness of the procedure and completes the proof.
\end{proof}

\section{About (very) small probabilities }\label{sec:smallprobs}
We now provide a tight lower bound for \cref{lem:double-exponential-patience}, by showing that in a constrained NE, players might need to use double-exponentially small probabilities.
This result itself is not surprising due to the $\exists\Rb$-completeness of the problem~\cite{HS20} even in games with $7$ players or more.
But we tighten this result by producing an explicit game  with only $5$ players where such double-exponentially small probabilities are required.
\begin{theorem}\label{thm:Gn}
    Let $n \in \Nb$.
    There exist two threshold vectors $\bx$ and $\by$ and a game $\Game$ where there is a stationary Nash equilibrium $\bsigma$ satisfying $\bx \leq \Eb(\bsigma) \leq \by$, but where all such stationary Nash equilibria contain probabilities that are double-exponentially small in $n$.
\end{theorem}

Let $n \in \Nb$.
The game $\Game^n$ is depicted by \Cref{fig:small_probas}.
It contains five players, named $\Circle$, $\Box$, $\triangle$, $\Diamond$, and $\pentagon$: the shape of each vertex indicates which player controls it, and black vertices are stochastic ones---the probability distribution is then indicated on the outgoing edges.
All omitted rewards are $0$.
For convenience, we allow rewards greater than $1$ here, accounting for the fact that all rewards can easily be normalised by dividing them by the greatest of them (namely $n$).
For now, the reader should ignore the notations $\beta_i$, $\gamma_i$, $\delta_i$, $\eta_i$.

Intuitively, the construction is the following: in each gadget $\sfdef_i$, from the vertex $r_i$, player $\Circle$ has to randomise between the terminal vertices $t_i$ and $t_i'$ (in the gadget $\sfdef_0$ the randomisation is imposed, because player $\Circle$ does not control the vertex $r_0$).
The probability of going to the vertex $t_i$ will be called $\alpha_i$.
We are interested in NEs such that player $\Circle$ gets expected payoff $1$: that means that the players should never go to those gadgets, nor to the terminal vertices $t_{i0}$ or $t_{i1}$.
But then, deviating and going to those gadgets or vertices must never be a profitable deviation. This imposes strict restrictions on which expected payoffs the players should get in each gadget $\mul_i$, and therefore on how they can randomise their own strategies.
In particular, from the vertex $c_i$, we will see that player $\Diamond$ must necessarily randomise, which is possible in a Nash equilibrium only if both sides are equivalent from her perspective: with this idea, each gadget $\mul_i$ binds the values of $\alpha_{i-1}$ and $\alpha_i$, imposing $\alpha_i = \alpha_{i-1}^2$.
Then, by induction, we find $\alpha_n = \frac{1}{2^{2^n}}$.

\begin{figure}[h]
	\begin{subfigure}[b]{0.4\textwidth}
	    \centering
		\begin{tikzpicture}
    		\node[stoch, initial above={}] (s0) {$s_0$};
    		\node[below of=s0] (d) {\dots};
            \node[left of=d, gadget] (mul2) {$\mul_2$};
            \node[left of=mul2, gadget] (mul1) {$\mul_1$};
            \node[right of=d, gadget] (muln) {$\mul_n$};
    
            \path[->, bend right] (s0) edge node[above left] {$\frac{1}{n}$} (mul1);
            \path[->, bend right] (s0) edge node[below right] {$\frac{1}{n}$} (mul2);
            \path[->, bend left] (s0) edge node[above right] {$\frac{1}{n}$} (muln);
		\end{tikzpicture}
		\caption{The game $\Game^n$}
		\label{fig:smalle_probas_complete}
	\end{subfigure}
	\begin{subfigure}[b]{0.25\textwidth}
	    \centering
		\begin{tikzpicture}
		      \node[stoch, initial above, initial text={}] (v0) {$r_0$};
            \node[below left of=v0, scale=0.6] (t) {$t_i:~\stackrel{\square}{1}~\stackrel{\smalltriangle}{(n-1)}~\stackrel{\diamond}{1}$};
            \node[below right of=v0, scale=0.6] (t') {$t_i':~\stackrel{\smalltriangle}{n}~\stackrel{\smallpentagon}{1}$};

            \path[->] (v0) edge node[above left] {$\frac{1}{2}$} (t);
            \path[->] (v0) edge node[above right] {$\frac{1}{2}$} (t');
		\end{tikzpicture}
	    \caption{The gadget $\sfdef_0$} \label{fig:small_probas_def0}
	\end{subfigure}
	\begin{subfigure}[b]{0.25\textwidth}
	    \centering
		\begin{tikzpicture}[scale=0.7]
		      \node[circlev, initial above, initial text={}] (vi) {$r_i$};
            \node[below left of=vi, scale=0.6] (t) {$t_i:~\stackrel{\square}{1}~\stackrel{\smalltriangle}{(n-i-1)}~\stackrel{\diamond}{1}$};
            \node[below right of=vi, scale=0.6] (t') {$t_i':~\stackrel{\smalltriangle}{(n-i)}~\stackrel{\smallpentagon}{1}$};
            \path[->] (vi) edge node[above left, gray] {$\alpha_i$} (t);
            \path[->] (vi) edge (t');
		\end{tikzpicture}
	    \caption{The gadget $\sfdef_i$} \label{fig:small_probas_defi}
	\end{subfigure}
	\begin{subfigure}[b]{\textwidth}
	    \centering
		\begin{tikzpicture}[scale=1,->]
		      \node[trianglev, initial left] (a) {$a_i$};
            \node[pentagonv, right of=a] (b) {$b_i$};
            \node[diamondv, right of=b] (c) {$c_i$};
            \node[squarev, above of=c] (d) {$d_i$};
            \node[gadget, above right of=d] (defi) {$\sfdef_i$};
            \node[trianglev, below right of=defi] (e) {$e_i$};
            \node[circlev, right of=e] (f) {$f_i$};
            \node[squarev, below of=c] (g) {$g_i$};
            \node[trianglev, right of=g] (h) {$h_i$};
            \node[circlev, above right of=h] (j) {$j_i$};
            \node[diamondv, below right of=j] (k) {$k_i$};
            \node[pentagonv, right of=k] (l) {$\ell_i$};
            \node[gadget, below right of=h] (defi-1) {$\sfdef_{i-1}$};
            \node[circlev, right of=l] (m) {$m_i$};
            \node[below of=a, scale=0.7] (t0) {$t_{i0}:~\stackrel{\smalltriangle}{n-i+\frac{1}{8}}$};
            \node[below of=b, scale=0.7] (t1) {$t_{i1}:~\stackrel{\smallpentagon}{\frac{11}{8}}$};
            \node[above right of=f, scale=0.7] (t2){$t_{i2}:~\stackrel{\circ}{1}~\stackrel{\square}{1}~\stackrel{\smalltriangle}{(n-i-1)}~\stackrel{\diamond}{1}~\stackrel{\smallpentagon}{2}$};
            \node[below right of=f, scale=0.7] (t3){$t_{i_3}:~\stackrel{\circ}{1}~\stackrel{\smalltriangle}{(n-i)}~\stackrel{\smallpentagon}{2}$};
            \node[above right of=m, scale=0.7] (t4){$t_{i4}:~\stackrel{\circ}{1}~\stackrel{\square}{1}~\stackrel{\smalltriangle}{(n-i)}~\stackrel{\diamond}{1}$};
            \node[below right of=m, scale=0.7] (t5){$t_{i5}:~\stackrel{\circ}{1}~\stackrel{\square}{1}~\stackrel{\smalltriangle}{(n-i)}~\stackrel{\smallpentagon}{1}$};
            \node[above left of=j, scale=0.7] (t6) {$t_{i6}:~\stackrel{\circ}{1}~\stackrel{\smalltriangle}{(n-i+1)}~\stackrel{\smallpentagon}{1}$};

            \path (a) edge (b);
            \path (a) edge (t0);
            \path (b) edge (c);
            \path (b) edge (t1);
            \path (c) edge node[gray, left] {$\beta_i$} (d);
            \path (d) edge (defi);
            \path (d) edge (e);
            \path (e) edge (defi);
            \path (e) edge (f);
            \path (f) edge node[gray, below right] {$\gamma_i$} (t2);
            \path (f) edge (t3);
            \path (c) edge (g);
            \path (g) edge (h);
            \path (h) edge (j);
            \path (j) edge node[gray, above right] {$\delta_i$} (k);
            \path (j) edge (t6);
            \path (k) edge (l);
            \path (g) edge[bend right] (defi-1);
            \path (h) edge (defi-1);
            \path (k) edge (defi-1);
            \path (l) edge[bend left] (defi-1);
            \path (l) edge (m);
            \path (m) edge node[gray, below right] {$\eta_i$} (t4);
            \path (m) edge (t5);
		\end{tikzpicture}
	    \caption{The gadget $\mul_i$} \label{fig:small_probas_muli}
	\end{subfigure}
    \caption{A game where very small probabilities are necessary. The rewards corresponding to all the players not explicitly mentioned are zero.} \label{fig:small_probas}
    \end{figure}
\begin{restatable}{lemma}{doubleexpsmall}\label{lem:doubleexpsmall} \label{lem:if_NE_then_double_exp}
    In the game $\Game^n$, if there exists a stationary Nash equilibrium $\bsigma$ in which player $\Circle$ gets expected payoff $1$, then we have $\sigma_\circ(r_n) = \frac{1}{2^{2^n}}$.
\end{restatable}

We now know that if a stationary NE satisfying this constraint exists, it necessarily includes a double-exponentially small probability.
We now prove that such an NE exists.

\begin{restatable}{lemma}{doubleExpNEExists}\label{lem:double_exp_NE_exists}
    For every $n \in \Nb$, the game $\Game^n$ has a stationary Nash equilibrium where player $\Circle$'s expected payoff is $1$.
\end{restatable}

Together, \cref{lem:doubleexpsmall,lem:double_exp_NE_exists} prove \Cref{thm:Gn}.
This theorem should be understood as a lower bound on \Cref{lem:double-exponential-patience}, that shows that one cannot expect a better complexity than $\NP^\NP$ using the techniques we use here.

However, let us now note that the above result is true only for exact Nash equilibria.
For $\epsilon$-NEs, our construction does not lead to the same result, as shown by the following lemma.

\begin{restatable}{lemma}{epsilonImpliesSimpleExp}\label{lem:epsilon_implies_simple_exp}
        Let $n \in \Nb$ and $\epsilon > 0$.
    In the game $\Game^n$, there exists a stationary $\epsilon$-NE $\btau$ where player $\Circle$'s expected payoff is $1$, and where for each player $i$ and every edge $uv$ with $u \in V_i$, the size $\bit\left(\tau_i(u)(v)\right)$ is bounded by a polynomial of $\bit(n)$ and $\bit(\epsilon)$.
    \end{restatable}

We conjecture that this can be generalised, and that as soon as we consider $\epsilon$-NEs instead of NEs, double-exponentially small probabilities are no longer required.

\begin{restatable}{conjecture}{expsmallprobs}\label{conjecture:1}
    Let $\Game$ be a game, let $\bx, \by$ be two vectors denoting the constraints, and let $\epsilon > 0$.
    If there exists a stationary $\epsilon$-NE in $\Game$ where the expected payoffs are between $\bx$ and $\by$, then there exists one that has polynomial sized fixed point representation.
\end{restatable}

Proving such a result would lead to an $\NP$ upperbound on our approximate problem, since such an $\epsilon$-NE can then be guessed and checked in polynomial time.

\section{A lower bound}\label{sec:lowerbounds}
We dedicate this section to show $\NP$-hardness of the approximate version of the constrained existence problem.

\begin{theorem}
    The approximate constrained existence problem of NEs is $\NP$-hard, even with two players\footnote{This result is stronger than the one presented in the conference version of this paper, where hardness was proven only for an arbitrary number of players.}.
\end{theorem}

\begin{proof}
    We proceed by reduction from the $3\SAT$ problem.
    Let $\phi = \bigvee_{i=1}^m \left(L_{i1} \wedge L_{i2} \wedge L_{i3}\right)$ be a CNF formula, over the variables $x_1, \dots, x_n$.

    \paragraph*{Construction of the game $\Game^\phi$}

    We define the game $\Game^\phi$ as follows.
    \begin{itemize}
        \item There are two players, player $\Circle$ and player $\Box$.

        \item Player $\Circle$ controls the vertex $\top_L$, for each literal $L$.
        Player $\Box$ controls the vertex $C_i$, for each clause $C_i$, and the vertex $L$, for each literal $L$.
        The stochastic vertices are named $s_0, s_1$, and $s_L$, for each literal $L$.
        The terminal vertices are named $t_\square$, $t_\emptyset$, $t_{\circ\square}$, and $t_\circ$.

        \item The initial vertex is the stochastic vertex $s_0$, from which the outgoing edges are distributed as follows:
        \begin{itemize}
            \item the probability of going to the vertex $x_k$, for each $k \in \{1, \dots, n\}$, is $\frac{1}{2^{k+1}}$;
            
            \item the probability of going to the vertex $\neg x_k$, for each $k$, is $\frac{1}{2^{k+1}}$;
            
            \item the probability of going to the vertex $s_1$ is $\frac{1}{2^{n+2}}$;
            
            \item the probability of going to the terminal vertex $t_\circ$ is $\frac{1}{2^{n+2}}$.
        \end{itemize}

        \item From the stochastic vertex $s_1$, the outgoing edges are distributed as follows: for each $i \in \{1, \dots, m\}$, the probability of going to the vertex $C_i$ is $\frac{1}{m}$.

        \item From each vertex $C_i$, player $\Box$ can move to the vertices $t_\square$, $L_{i1}$, $L_{i2}$, and $L_{i_3}$.

        \item From each vertex $L$, player $\Box$ can move to the vertices $\top_L$ and $t_\circ$.

        \item From each vertex $\top_L$, player $\Circle$ can move to the vertices $t_{\circ\square}$ or $t_\emptyset$.

        \item In the terminal vertex $t_\square$, player $\Circle$ gets the payoff $0$ and player $\Box$ gets the payoff $1$.
        In the terminal vertex $t_\emptyset$, both players receive the payoff $0$.
        In the terminal vertex $t_{\circ\square}$, they both receive $1$, and in $t_\circ$, player $\Circle$ receives $1$, and player $\Box$ receives $0$.
    \end{itemize}

    \begin{figure}
        \centering
        \begin{tikzpicture}[every node/.style={scale=0.7}, node distance=2.2cm,initial text={}]
    		\node[stoch, initial] (s0) {$s_0$};
            \node[below left of=s0] (t0) {$t_\circ:~\stackrel{\circ}{1}\stackrel{\square}{0}$};
    		\node[below of=s0, stoch] (s1) {$s_1$};
            \path[->] (s0) edge[bend left=10] node[above left] {$\frac{1}{2^{n+2}}$} (t0);
            \path[->] (s0) edge[bend left=10] node[right] {$\frac{1}{2^{n+2}}$} (s1);

            \node[right of=s1, squarev] (C1) {$C_1$};
            \node[below of=C1] (dots1) {\dots};
            \node[below of=dots1, squarev] (Cm) {$C_m$};
            \node[left of=dots1] (t1) {$t_\square:~\stackrel{\circ}{0}\stackrel{\square}{1}$};
            \path[->] (s1) edge[bend left=10] node[above] {$\frac{1}{m}$} (C1);
            \path[->] (s1) edge[bend left=10] node[above right] {$\frac{1}{m}$} (Cm);
            \path[->] (C1) edge (t1);
            \path[->] (Cm) edge (t1);

            \node[right of=C1] (fake) {};
            \node[right of=fake, squarev] (xn) {$x_n$};
            \node[below of=xn, squarev] (nxn) {$\neg x_n$};
            \node[above of=xn] (dots2) {\dots};
            \node[squarev, above of=dots2] (nx1) {$\neg x_1$};
            \node[squarev, above of=nx1] (x1) {$x_1$};
            \path[->] (s0) edge[bend left=10] node[above left] {$\frac{1}{4}$} (x1);
            \path[->] (s0) edge[bend left=10] node[above left] {$\frac{1}{4}$} (nx1);
            \path[->] (s0) edge[bend left=10] node[above right] {$\frac{1}{2^{n+1}}$} (xn);
            \path[->] (s0) edge[bend left=10] node[above right] {$\frac{1}{2^{n+1}}$} (nxn);
            \path[->, gray, dashed] (C1) edge (x1);
            \path[->, gray, dashed] (C1) edge (xn);
            \path[->, gray, dashed] (C1) edge (nxn);

            \node[circlev, above right of=xn] (topxn) {$\top_{x_n}$};
            \node[below right of=xn] (tbot) {$t_\circ:~\stackrel{\circ}{1}\stackrel{\square}{0}$};
            \node[above right of=topxn] (ttop) {$t_{\circ\square}:~\stackrel{\circ}{1}\stackrel{\square}{1}$};
            \node[below right of=topxn] (t2) {$t_\emptyset:~\stackrel{\circ}{0}\stackrel{\square}{0}$};
            \path[->] (xn) edge (topxn);
            \path[->] (xn) edge (tbot);
            \path[->] (topxn) edge (ttop);
            \path[->] (topxn) edge (t2);

            \node[above of=xn] (dots3) {\dots};
            \node[right of=x1] (dots4) {\dots};
            \node[right of=nx1] (dots5) {\dots};
            \node[right of=Cm] (dots6) {\dots};
            \node[right of=nxn] (dots7) {\dots};
            
		\end{tikzpicture}
        \caption{A reduction from $3\SAT$}
        \label{fig:sat}
    \end{figure}

    That game is depicted by \Cref{fig:sat}.
    The grey dashed arrows correspond to the case where $C_1 = x_1 \vee x_n \vee \neg x_n$.
    For convenience, the vertex $t_\circ$ is depicted twice.

    We now state the following result, whose proof is found in \cref{app:LB}.

    \begin{restatable}{proposition}{formulaimpliesNE}\label{prop:formulaimpliesNE}
        If the formula $\phi$ is satisfiable, then there exists a stationary Nash equilibrium in the game $\Game^\phi$ such that player $\Circle$'s expected payoff is $1$, and player $\Box$'s expected payoff $\frac{1}{2}$.
    \end{restatable}

    For the converse, we show this stronger result, whose proof is also in \cref{app:LB}.

    \begin{restatable}{proposition}{NEimpliesFormula}\label{prop:NEimpliesFormula}
        If there is a stationary $\epsilon$-Nash equilibrium $\bsigma$ in the game $\Game^\phi$ that satisfy the constraints:
        \begin{itemize}
            \item $\Eb(\bsigma)[\mu_\circ] \geq 1 - \epsilon$;
            \item $\frac{1}{2} - \epsilon \leq \Eb(\bsigma)[\mu_\square] \leq \frac{1}{2} + \epsilon$.
        \end{itemize}
        where:
        $$\epsilon = \frac{1}{2^{n+3} (6n + 8m + 1)},$$
        then the formula $\phi$ is satisfiable.
    \end{restatable}

    Together, these two propositions prove that our promise problem is $\NP$-hard.
\end{proof}

\section{Conclusion and discussion}
We studied the (functional) approximate constrained problem of NEs. Specifically, we asked the following: assuming the existence of a stationary Nash equilibrium that satisfies given payoff constraints, can we efficiently compute a stationary strategy profile that is an $\epsilon$-Nash equilibrium and satisfies the constraints up to an additive error $\epsilon$?

Although our results focus on turn-based stochastic games, we note that the lower bound also holds in more general settings, including concurrent games (games on graphs where the players select actions simultaneously) and broader objectives such as limit-average (mean-payoff) objectives. 
For the upper bound, we remark that the existential theory of reals formulations we use for turn-based games with reachability objectives can be naturally adapted to concurrent games with mean-payoff objectives, with only a constant increase in the degree of the polynomials involved. By following similar proof techniques, one can still obtain an $\FNP^\NP$ upper bound for the corresponding constrained $\epsilon$-NE problem in these more general settings. However, we restrict our presentation and main results to turn-based games with reachability objectives in this paper to maintain clarity.

We showed that the decision version of this problem lies in the class $\NP^\NP$, and the corresponding search problem belongs to $\FNP^\NP$. On the other hand, we established only an $\NP$-hardness lower bound. This leaves a gap in the known complexity of the problem, and closing this gap remains an open question.

One possible direction is to prove that if an $\epsilon$-Nash equilibrium exists, then there also exists one in which all probabilities are at least inverse-exponential in the size of the input, assuming the existence of a Nash equilibrium. Such a result would not contradict our example in Section~\ref{sec:smallprobs}, which shows that double-exponentially small probabilities are required in the exact version of the problem. If exponentially small probabilities suffice in the approximate setting, then a strategy profile could be represented using only polynomially many bits in the fixed-point representation. Since such a strategy yields a polynomial certificate, and verification requires only finding the values of MDPs with fixed-point representation, this would yield an $\NP$-completeness result for the decision version.

Alternatively, if such a bound does not hold, and there exists a counterexample where all $\epsilon$-Nash equilibria satisfying the constraints require double-exponentially small probabilities, this would suggest an inherent complexity. It would also point to intrinsic limitations in our techniques, analogous to those demonstrated by Hansen, Koucký, and Miltersen~\cite{HKM09}, who showed that double-exponentially small probabilities are necessary for $\epsilon$-optimal strategies in concurrent stochastic games. Their result implies that any algorithm that explicitly manipulates these probabilities must use at least exponential space in the worst case, and our results would also imply the same for our problem. 

We end by recalling our conjecture in \cref{sec:smallprobs}.
\expsmallprobs*
\bibliography{biblio}
 \appendix

\section{Appendix for \cref{sec:smallprobs}}\label{app:smallprobs}
\subsection{Proof of \cref{lem:doubleexpsmall}}

\doubleexpsmall*
\begin{proof}
    Let $\bsigma$ be such a stationary NE.

    For each $i \in \{0, \dots, n\}$, we write $\alpha_i, \beta_i, \gamma_i, \delta_i$, and $\eta_i$, for the probabilities $\sigma_\circ(r_i)(t)$, $\sigma_\diamond(c_i)(d_i)$, $\sigma_\circ(f_i)(t_{i2})$, $\sigma_\circ(j_i)(k_i)$, and $\sigma_\circ(m_i)(t_{i4})$, respectively.
    Since the vertex $r_0$ is stochastic, we fix $\alpha_0 = \frac{1}{2}$.
    We will now show the following.

    \begin{proposition}
        For each $i$, we have $\alpha_i = \frac{1}{2^{2^i}}$.
    \end{proposition}

    \begin{proof}
        The proposition is immediately true for $i = 0$.
    
        Let now $i > 0$, and let us assume that we have $\alpha_{i-1} = \frac{1}{2^{2^{i-1}}}$
        Let us already note that since player $\Circle$ gets expected payoff $1$, and since there is no terminal vertex in this game where she can get more than $1$, all terminal vertices where she gets less than $1$ must be reached with probability $0$.
        That implies for example that from the vertex $a_i$, which is reached with non-zero probability, player $\triangle$ goes deterministically to $b_i$, and that from the vertex $b_i$, player $\pentagon$ goes deterministically to the vertex $c_i$.

        This will first help us prove the following claim.

        \begin{claim}\label{claim:beta}
            We have $0 < \beta_i < 1$.
        \end{claim}

        \begin{proof}
            Since player $\triangle$ and player $\pentagon$ go to the terminal vertices $t_{i0}$ and $t_{i1}$, respectively, with probability $0$, and have no profitable deviation in $\bsigma$, we know that from the vertex $c_i$, those players must get at least expected payoffs $n-i+\frac{1}{8}$ and $\frac{11}{8}$, respectively.
            For player $\triangle$, that implies $\beta_i < 1$, since player $\triangle$'s expected payoff is always smaller than or equal to $n-i$ if the vertex $d_i$ is reached.
            For player $\pentagon$ it implies $\beta_i > 0$, since player $\pentagon$'s expected payoff is smaller than or equal to $1$ if the vertex $g_i$ is reached.
        \end{proof}
        
        Now, since the vertex $d_i$ is reached with non-zero probability, and since no terminal vertex giving player $\Circle$ the reward $0$ must ever be reached, we necessarily have $\sigma_\square(d_i)(r_i) = \sigma_{\smalltriangle}(e_i)(r_i) = 0$, and from the vertex $d_i$ player $\Box$ and $\triangle$ deterministically go to the vertex $f_i$.
        There, we can prove the following claim.

        \begin{claim}\label{claim:gamma}
            We have $\gamma_i = \alpha_i$.
        \end{claim}
        
        \begin{proof}
            By deviating and going to the gadget $\sfdef_i$, player $\Box$ and player $\triangle$ can get the expected payoffs $\alpha_i$ and $n-i-\alpha_i$, respectively.
            But none of those deviations is profitable, since the strategy profile $\bsigma$ is an NE.
            For player $\Box$, this means that he gets at least the expected payoff $\alpha_i$ from $f_i$, i.e., that we have $\gamma_i \geq \alpha_i$.
            For player $\triangle$, this means that she gets at least the expected payoff $n-i-\alpha_i$ from $f_i$, i.e., that we have $\gamma_i \leq \alpha_i$.
        \end{proof}      

        On the other side, the same argument implies that player $\Box$ and $\triangle$ deterministically go to the vertex $j_i$, and that we have, then, the following result:

        \begin{claim}\label{claim:delta}
            We have $\delta_i = \alpha_{i-1}$.
        \end{claim}

        \begin{proof}
            Since going to the gadget $\sfdef_{i-1}$ is not a profitable deviation for player $\Box$, he must get at least the expected payoff $\alpha_{i-1}$ by going to $j_i$, which implies $\delta_i \geq \alpha_{i-1}$.
            Since we have $\alpha_{i-1} = \frac{1}{2^{2^i}} > 0$ by induction hypothesis, we now know that the vertex $k_i$ is reached with positive probability, which implies that from there, player $\diamond$, and then player $\pentagon$, deterministically go to the vertex $m_i$.
            Thus, if the vertex $k_i$ is reached, player $\triangle$ gets expected payoff $n-i$.
            Since going to the gadget $\sfdef_{i-1}$ is not profitable for her either, she must get at least payoff $n-i+1-\alpha_{i-1}$ by going to $j_i$, that is, we have $\delta_i \leq \alpha_{i-1}$.
        \end{proof}

        Now that we know that the vertex $k_i$ is reached with positive probability, and therefore that from there, player $\Diamond$ and $\pentagon$ go to the vertex $m_i$, we can apply the same reasoning in that last vertex.

        \begin{claim}\label{claim:eta}
            We have $\eta_i = \alpha_{i-1}$.
        \end{claim}

        \begin{proof}
            Since going to the gadget $\sfdef_{i-1}$ is not a profitable deviation for player $\Diamond$, we have $\eta_i \geq \alpha_{i-1}$.
            Since it is not a profitable deviation for player $\pentagon$, we have $\eta_i \leq \alpha_{i-1}$.
        \end{proof}

        We can now connect $\alpha_{i-1}$ and $\alpha_i$.
        By \Cref{claim:beta}, from the vertex $c_i$, player $\Diamond$ proceeds to a strict randomisation.
        Since going deterministically to $d_i$ or to $g_i$ must not be profitable deviations, he must get exactly the same expected payoff on both sides.
        From $d_i$, using \Cref{claim:gamma}, we find that he gets expected payoff $\alpha_i$.
        From $g_i$, using \Cref{claim:delta} and \Cref{claim:eta}, we find that he gets expected payoff $\alpha_{i-1}^2$.
        Using the induction hypothesis, we obtain:
        $$\alpha_i = \alpha_{i-1}^2 = \frac{1}{2^{2^i}},$$
        as desired.
    \end{proof}

    By taking $i = n$, we can conclude.
\end{proof}

\subsection{Proof of \cref{lem:double_exp_NE_exists}}

\doubleExpNEExists*

\begin{proof}
    The proof of \Cref{lem:if_NE_then_double_exp} defines almost completely how such a stationary NE must behave.
    The only probability that was not completely defined was $\beta_i$: let us fix it to $\frac{1}{2}$, and let us call this strategy profile $\bsigma$.
    Let us now check that no player has a profitable deviation.

    \subp{Player $\Circle$}
    Player $\Circle$ has expected payoff $1$ when following $\bsigma$, and there is no terminal vertex that could give her a better reward: therefore, she has no profitable deviation.

    \subp{Player $\Box$}
    From each vertex of the form $d_i$, the play continues deterministically to $f_i$ and then ends in $t_{i2}$ with probability $\gamma_i = \alpha_i$, giving player $\Box$ expected payoff $\alpha_i$.
    Thus, going to the gadget $\sfdef_i$ is not a profitable deviation.
    From vertices of the form $g_i$, the situation is analogous.

    \subp{Player $\triangle$}
    From each vertex of the form $e_i$, player $\triangle$ gets expected payoff $n-i-\alpha_i$, hence going to the gadget $\sfdef_i$ is not a profitable deviation; and the situation is analogous from the vertices of the form $h_i$.
    From each vertex $a_i$, she gets expected payoff:
    $$\frac{1}{2}(n-i-\alpha_i) + \frac{1}{2}(n-i+1-\alpha_{i-1})$$
    $$= n-i + \frac{1}{2} -\frac{1}{2}\alpha_i -\frac{1}{2} \alpha_{i-1}$$
    Since $\alpha_{i-1} \leq \alpha_0 = \frac{1}{2}$ and $\alpha_i \leq \alpha_1 = \frac{1}{4}$, this expected payoff is greater than or equal to the quantity:
    $$n-i + \frac{1}{2} - \frac{1}{2} \frac{1}{4} - \frac{1}{2} \frac{1}{2}$$
    $$= n-i + \frac{1}{8},$$
    hence going to the terminal vertex $t_{i0}$ is not a profitable deviation either.

    \subp{Player $\Diamond$}
    From each vertex of the form $k_i$, player $\Diamond$ gets expected payoff $\alpha_{i-1}$, hence going to the gadget $\sfdef_{i-1}$ is not a profitable deviation.
    From the vertex $c_i$, player $\Diamond$ will get expected payoff $\alpha_i = \alpha_{i-1}^2$ whether she goes to $d_i$ or to $g_i$, hence she has no profitable deviation from $c_i$ either.

    \subp{Player $\pentagon$}
    From every vertex of the form $\ell_i$, player $\pentagon$ gets expected payoff $1-\alpha_{i-1}$, hence going to the gadget $\sfdef_{i-1}$ is not a profitable deviation.
    From a vertex of the form $b_i$, player $\pentagon$ gets expected payoff:
    $$\frac{1}{2} 2 + \frac{1}{2} \left( 1-\alpha_{i-1} + \alpha_{i-1} (1- \alpha_{i-1}) \right)$$
    $$= \frac{3}{2} - \frac{1}{2} \alpha_{i-1}^2$$
    $$\geq \frac{3}{2} - \frac{1}{2} \left( \frac{1}{2} \right)^2$$
    $$= \frac{11}{8},$$
    hence she has no profitable deviation from the vertex $b_i$ either.
\end{proof}

\subsection{Proof of \cref{lem:epsilon_implies_simple_exp}}

\epsilonImpliesSimpleExp*

\begin{proof}
    Let $I$ be the smallest integer such that we have:
    $$\frac{1}{2^{2^I+1} n} \leq \epsilon.$$
    Let us consider the strategy profile $\btau$, where, for every $i$:
    \begin{itemize}
        \item from the vertex $r_i$, player $\Circle$ goes to $t_i$ with probability $\theta_i$;

        \item from the vertex $a_i$, the players go deterministically to $c_i$, where player $\Diamond$ randomises uniformly between $d_i$ and $g_i$;

        \item from the vertex $d_i$, the players go deterministically to $f_i$, where player $\Circle$ goes to $t_{i2}$ with probability $\theta_i$;

        \item from the vertex $g_i$, the players go deterministically to the vertex $j_i$, where player $\Circle$ goes to $k_i$ with probability $\theta_{i-1}$;

        \item from the vertex $k_i$, the players go deterministically to the vertex $m_i$, where player $\Circle$ goes to $t_{i4}$ with probability $\theta_{i-1}$;
    \end{itemize}
    where $\theta_i = \frac{1}{2^{2^i}}$ for $i \leq I$, and $\theta_i = 0$ for $i > I$.

    Following the same arguments as in the proof of \Cref{lem:double_exp_NE_exists}, no player has a profitable deviation anywhere in this strategy profile, except player $\Diamond$, in the gadget $\mul_I$, who now has an incentive to go deterministically to the vertex $g_I$.
    Indeed, if the vertex $d_I$ is reached, he gets expected payoff $0$, while if $g_I$ is reached, he gets $\theta_{I-1}^2 = \frac{1}{2^{2^I}}$.
    Then, since the vertex $c_I$ is itself reached with probability $\frac{1}{n}$, such a deviation is profitable by:
    $$\frac{1}{n} \frac{1}{2} \frac{1}{2^{2^I}} = \frac{1}{2^{2^I + 1} n},$$
    which is smaller than $\epsilon$ by hypothesis.
    Therefore, the strategy profile $\btau$ is an $\epsilon$-NE, and all probabilities can be written with a number of bits that is polynomial in the size of the instance (which includes $\epsilon$).
\end{proof}

\section{Appendix for \cref{sec:lowerbounds}}\label{app:LB}
\formulaimpliesNE*
\begin{proof}[Proof of \cref{prop:formulaimpliesNE}]
We first define the strategy profile $\bsigma$, then show that it satisfies the constraints, and finally that it is an NE.

    \subparagraph*{Definition of $\bsigma$}
        Let $\nu: X \to \{0, 1\}$ be a valuation satisfying the formula $\phi$.
        We define the stationary strategy profile $\bsigma$ as follows: from each clause vertex $C_i$, player $\Box$ goes deterministically to a given vertex $L$ such that the literal $L$ is satisfied by the valuation $\nu$.
        From each vertex of the form $L$, player $\Box$ goes deterministically to the vertex $\top_{L}$ if the valuation $\nu$ satisfies $L$, and to the vertex $t_\bot$ otherwise.
        From each vertex of the form $\top_L$, player $\Circle$ deterministically goes to the vertex $t_\top$ if the valuation $\nu$ satisfies $L$, and to the vertex $t_2$ otherwise.
        
        \subparagraph*{The strategy profile $\bsigma$ satisfies the constraints.}
        First, let us note that this game contains no cycle, hence the probability of never reaching any terminal vertex is $0$.

        This strategy profile never reaches the terminal vertex $t_1$, nor $t_2$.
        Therefore, player $\Circle$'s expected payoff is $1$.

        As concerns player $\Box$, he gets the payoff $0$ only when the vertex $t_\bot$ is reached, i.e. only in plays that either reach $t_\bot$ immediately from $s_0$, or that traverse a vertex $L$ such that the literal $L$ is not satisfied.
        For each $k$, exactly one of the two literals $x_k$ and $\neg x_k$ is in that situation.
        That literal is reached with probability exactly $\frac{1}{2^{k+1}}$ (it cannot be reached from $s_1$, since player $\Box$'s strategy targets only satisfied literals).
        Moreover, the probability of reaching $t_\bot$ immediately is $\frac{1}{2^{k+2}}$.
        Therefore, the probability of reaching that vertex is $\frac{1}{4} + \dots + \frac{1}{2^{n+1}} + \frac{1}{2^{n+2}} = \frac{1}{2}$, and player $\Box$'s expected payoff is $\frac{1}{2}$.

        \subparagraph*{The strategy profile $\bsigma$ is an NE.}
        Player $\Circle$ has no profitable deviation, since she gets the expected payoff $1$, and could not get more.
        As for player $\Box$, from any vertex of the form $C_i$, or $L$ where the literal $L$ is satisfied by $\nu$, he also gets the expected payoff $1$.
        From a vertex of the form $L$ where $L$ is not satisfied by $\nu$, he gets the payoff $0$, but the only deviation available consists in going to the vertex $\top_L$, from which player $\Circle$ goes to the vertex $t_2$, giving him also the payoff $0$.
\end{proof}

\NEimpliesFormula*
\begin{proof}[Proof of \cref{prop:NEimpliesFormula}]
        Let $\bsigma$ be such a stationary NE.
        For each vertex $v$ and each player $i \in \{\Circle, \Box\}$, we write $e_i(v)$ for player $i$'s expected payoff when the strategy profile $\bsigma$ is played from the vertex $v$.

        Let us first study the behavior of $\bsigma$ from a literal vertex.

        \begin{claim}
            Let $L$ be a literal, either of the form $x_k$ or of the form $\neg x_k$.
            Let $\delta = 2^{k+1}\epsilon$.
            Then, we have either $e_\square(L) \leq \delta$, or $e_\square(L) \geq 1 - 3\delta$.
        \end{claim}

        \begin{proof}
            The vertex $L$ is reached with probability at least $\frac{1}{2^{k+1}}$.
            Since no player has a deviation that is profitable by $\epsilon$ or more in the whole strategy profile, no player has a deviation that is profitable by $\delta$ or more from the vertex $L$.
            Similarly, from the vertex $L$, player $\Circle$ gets expected payoff at least $1-\delta$.
            Let us now write $\beta$ for the probability $\sigma_\square(L)(\top_L)$, and $\gamma$ for the probability $\sigma_\circ(\top_L)(t_\top)$.
            
            We have $e_\circ(L) = 1 - \beta + \beta\gamma$.
            That expected payoff must be larger than or equal to $1-\delta$, which implies:
            \begin{equation}\label{eq:3}
                \beta - \beta\gamma \leq \delta.
            \end{equation}
    
            Similarly, we have $e_\square(L) = \beta\gamma$.
            Since the deviation that consists in going deterministically to $\top_L$ is not profitable by more than $\delta$, we obtain:
            \begin{equation}\label{eq:4}
                \gamma - \beta\gamma \leq \delta.
            \end{equation}
    
            We can now separate two cases (note that those computations are valid because $\delta < \frac{1}{4}$):
            \begin{itemize}
                \item If $\beta \leq \frac{1}{2}$, then \Cref{eq:3} implies $\gamma - \frac{1}{2} \leq \delta$, i.e. $\gamma \leq 2\delta$, which implies:
                $$e_\square(L) = \beta\gamma \leq \frac{1}{2} 2\delta = \delta.$$
    
                \item If $\beta > \frac{1}{2}$, then \Cref{eq:3} implies $\frac{1}{2}(1-\gamma) \leq \delta$, i.e. $\gamma \geq 1 - 2\delta$, which implies $e_\square(L) \geq \frac{1}{2} - \delta$.
                Now, using \Cref{eq:4}, we obtain $(1-2\delta)(1-\beta) \leq \delta$, i.e.:
                $$\beta \geq \frac{1-3\delta}{1-2\delta}$$
                and finally:
                $$e_\square(L) = \beta\gamma \geq (1-2\delta) \frac{1 - 3\delta}{1 - 2\delta} = 1 - 3\delta.$$
            \end{itemize}
        \end{proof}
        
        Let us now call \emph{green} the literals $L$ such that $e_\square(L) \leq \delta$, and \emph{red} those such that $e_\square(L) \geq 1 - 3\delta$.
        As one can expect, they will correspond to true and false literals---but we still need to prove that that every clause has a green literal, and that opposite literals have different colors.

    \begin{claim}
        Every clause has a green literal.
    \end{claim}

    \begin{proof}
        Let $C_i$ be a clause of $\phi$. 
        The vertex $C_i$ is reached with probability $\frac{1}{2^{n+2}m}$.
        Therefore, we have $e_\circ(C_i) \geq 1 - \delta$, which implies $\sigma_\square(C_i)(t_1) \leq 2^{m+2}m\epsilon$.
    
        Now, since going deterministically to the vertex $t_1$ must not be profitable by more than $2^{m+2}m\epsilon$ for player $\Box$, we necessarily have $e_\square(C_i) \geq 1 - 2^{m+2}m\epsilon$.
    
        Moreover, there is necessarily one successor of the vertex $C_i$ to which player $\Box$ goes with probability at least $\frac{1}{4}$.
        Since we have shown $\sigma_\square(C_i)(t_1) \leq 2^{m+2}m\epsilon < \frac{1}{4}$, that successor is a literal vertex $L$.
        Then, we necessarily have $e_\square(L) \geq 1 - 2^{n+4} m \epsilon$, which implies that the literal $L$ is green.
    \end{proof}

    \begin{claim}
        Two opposite literals always have different colors.
    \end{claim}

    \begin{proof}
        We can now consider a general expression of player $\Box$'s expected payoff, and the following inequality:
        $$\frac{1}{2} - \epsilon \leq \sum_{k=1}^n \frac{1}{2^{k+1}} (e_\square(x_k) + e_\square(\neg x_k)) + \frac{1}{2^{n+2}} e_\square(s_1) \leq \frac{1}{2} + \epsilon.$$
        If all the quantities $e_\square(L)$ and $e_\square(s_1)$ are equal to either $0$ or $1$, the sum above ranges over the set:
        $$\left\{\left.\frac{p}{2^{n+2}} ~\right|~ p \in \{0, \dots, 2^{n+1} - 1\}\right\},$$
        and it reaches the quantity $\frac{1}{2}$ if and only if $e_\square(s_1) = 1$ and $e_\square(x_k) = 1 - e_\square(\neg x_k)$ for each $k$.
        In that case, we indeed have that two opposite literals always have different colors.

        Now, since the only constraints we have are $e_\square(L) \leq 2^{k+1} \epsilon$ or $e_\square(L) \geq 1 - 3 \times 2^{k+1} \epsilon$ for each $k$ and each $L \in \{x_k, \neg x_k\}$, and $e\square(s_1) \geq 1 - 2^{n+2}m \epsilon$, the error compared to the previously described situation is at most:
        $$\sum_{k=1}^n \frac{1}{2^{k+1}} (3 \times 2^{k+1} + 3 \times 2^{k+1}) + \frac{1}{2^{n+2}} 2^{n+2}m \epsilon$$
        $$= (6 n + m) \epsilon.$$
        Since we have $(6n + m) \epsilon + \epsilon < \frac{1}{2^{n+2}}$, this error is not sufficient to fall into the interval $\left[\frac{1}{2} - \epsilon, \frac{1}{2} + \epsilon\right]$ from a situation in which two opposite literals would have the same colour.
    \end{proof}

    Thus, the valuation that satisfies every green literal is a valuation that satisfies the formula $\phi$, which ends the proof.
\end{proof}

 \section{NP-hardness with $n$ players}
 In this section we give the reduction as in the conference version and its proof of correctness. We first recall the construction there and then the proof.

    We proceed by reduction from the $3\SAT$ problem.
    Let $\phi = \bigvee_{i=1}^m \left(L_{i1} \wedge L_{i2} \wedge L_{i3}\right)$ be a CNF formula, over the variable set $X$.

    \paragraph*{Construction of the game $\Game^\phi$}

    We define the game $\Game^\phi$ as follows.
    \begin{itemize}
        \item There is a player called \emph{Solver} and written $\solver$, and for each variable $x \in X$, there is a player $x$.

        \item For each clause $C_i$, there is a vertex $C_i$ controlled by Solver.
        The initial vertex is $C_1$.

        \item From each clause vertex $C_i$, Solver can move to the vertices $(C_i, L_{i1})$, $(C_i, L_{i2})$, and $(C_i, L_{i3})$.

        \item Each vertex of the form $(C_i, x)$ (positive literals) is a stochastic vertex.
        From there, there are two edges, each taken with probability $\frac{1}{2}$: one to the terminal vertex $t_x$, and one to the next clause vertex $C_{i+1}$ (or the vertex $s$ if $i=m$).

        \item Each vertex of the form $(C_i, \neg x)$ is controlled by player $x$.
        From there, player $x$ can move either to the terminal vertex $t_\solver$, or to the next clause vertex $C_{i+1}$ (or$s$ if $i=m$).

        \item The vertex $s$ is a stochastic vertex with two outgoing edges, each with probability $\frac{1}{2}$: one to the vertex $C_1$, one to the terminal vertex $t$.

        \item In the terminal vertex $t_i$, player $i$ receives the reward $0$, and every other player receives the reward $1$.
        In the terminal vertex $t$, every player receives the reward $1$.
    \end{itemize}

    \begin{figure}
        \centering
        \begin{tikzpicture}
    		\node[circlev, initial] (C1) {$C_1$};
    		\node[right of=C1, stoch] (x2) {$x_2$};
            \node[above of=x2, stoch] (x1) {$x_1$};
            \node[below of=x2, squarev] (nx3) {$\neg x_3$};
            \node[above of=x1] (tx1) {$t_{x_1}:~\stackrel{x_1}{0}\stackrel{\forall}{1}$};
            \node[below right of=tx1] (tx2) {$t_{x_2}:~\stackrel{x_2}{0}\stackrel{\forall}{1}$};
            \node[below of=nx3] (ts) {$t_\solver:~\stackrel{\solver}{0}\stackrel{\forall}{1}$};
            \node[right of=x2, circlev] (C2) {$C_2$};
            \node[right of=C2] (dots) {$\dots$};
            \node[right of=dots, stoch] (s) {$s$};
            \node[right of=s] (t) {$t:~\stackrel{\forall}{1}$};
    
            \path[->] (C1) edge (x1);
            \path[->] (C1) edge (x2);
            \path[->] (C1) edge (nx3);
            \path[->] (x1) edge (tx1);
            \path[->, bend right] (x2) edge (tx2);
            \path[->] (nx3) edge (ts);
            \path[->] (x1) edge (C2);
            \path[->] (x2) edge (C2);
            \path[->] (nx3) edge (C2);
            \path[->] (C2) edge (dots);
            \path[->] (dots) edge (s);
            \path[->] (s) edge (t);
            \path[->, bend left] (s) edge (C1);
		\end{tikzpicture}
        \caption{A reduction from $3\SAT$}
        \label{fig:sat}
    \end{figure}

    That game is depicted by \Cref{fig:sat}, in a case where $C_1 = x_1 \vee x_2 \vee \neg x_3$.
    Solver controls the circle vertices, and the square vertex is controlled by player $x_3$.
    We simply write $L$ for every vertex of the form $(C, L)$.
    The symbol $\forall$ should be interpreted as \emph{``every (other) player''}.
    Outgoing edges of stochastic vertices all have probability $\frac{1}{2}$.

    \paragraph*{Nash equilibria from a satisfiable formula}

    We now state and prove the following result.

    \begin{proposition}
        If the formula $\phi$ is satisfiable, then there exists a Nash equilibrium in the game $\Game^\phi$ such that Solver's expected payoff is $1$.
    \end{proposition}
    \begin{claimproof}
We first define the strategy profile $\bsigma$.

    \subparagraph*{Definition of $\bsigma$}
        Let $\nu: X \to \{0, 1\}$ be a valuation satisfying the formula $\phi$.
        We define the stationary strategy profile $\bsigma$ as follows: from each clause vertex $C_i$, Solver goes deterministically to a given vertex $(C_i, L)$ such that the literal $L$ is satisfied by the valuation $\nu$.
        From each vertex of the form $(C_i, \neg x)$, player $x$ goes deterministically to the vertex $C_{i+1}$ (or $s$).
        
        \subparagraph*{The strategy profile $\bsigma$ is an NE and satisfies the constraint.}
        First, let us note that since all cycles in this game traverse at least the stochastic vertex $s$, the probability of never reaching any terminal vertex is $0$.

        The probability of reaching the terminal vertex $t_\solver$ is also $0$.
        Thus, Solver gets expected payoff $1$, and has no profitable deviation, since no terminal vertex gives her a better reward.

        Let us now consider a variable player $x$.
        If we have $\nu(x) = 1$, then no vertex of the form $(C, \neg x)$ is ever visited, by definition of Solver's strategy.
        Thus, player $x$ is never given the opportunity to deviate.
        If we have $\nu(x) = 0$, then no vertex of the form $(C, x)$ is visited, and consequently, the probability of reaching the vertex $t_x$ is $0$, which means that player $x$'s expected payoff is $1$ and that such player has no profitable deviation, which ends the proof.
        \end{claimproof}

        \paragraph*{Non-existence of $\epsilon$-Nash equilibria from unsatisfiable formula}

    For the converse, we show the following stronger result.

    \begin{proposition}
        For $m$ large enough, if the formula $\phi$ is not satisfiable, then there is no $\epsilon$-Nash equilibrium in the game $\Game^\phi$ such that Solver's expected payoff is greater than or equal to $1-\epsilon$, where $\epsilon = 2^{-3m}$.
    \end{proposition}
        \begin{proof}
        Let us assume that there is no valuation satisfying the formula $\phi$.
        Let $\bsigma$ be a stationary strategy profile in the game $\Game^\phi$, such that Solver's expected payoff is at least $1-\epsilon$: let us show that $\bsigma$ is not an $\epsilon$-NE.

        Let $\nu$ be the valuation defined as follows: for every variable $x$, we have $\nu(x) = 1$ if and only if there is a vertex of the form $(C, x)$ such that $\sigma_\solver(C)(C, x) \geq \frac{1}{3}$.
        Since $\phi$ is not satisfiable, there is a clause $C_i$ that is not satisfied by $\nu$.
        On the other hand, since the vertex $C_i$ has only three outgoing edges, there is at least one literal $L$ such that $\sigma_\solver(C_i)(C_i, L) \geq \frac{1}{3}$.
        Let us pick one.

\subp{The literal $L$ is necessarily negative}
        If $L = x$ is a positive literal, then by definition of $\nu$, we have $\nu(x) = 1$, which means that the valuation $\nu$ actually satisfies the clause $C_i$: this case is therefore impossible.
        Therefore, there is a variable $x$ such that $L = \neg x$.

\subp{A deviation for player $x$}
Since the valuation $\nu$ does not satisfies the clause $C_i$, we have $\nu(x) = 1$, which means that there is a clause $j$ such that $\sigma_\solver(C_j)(C_j, x) \geq \frac{1}{3}$.
Then, from the vertex $(C_j, x)$, with probability $\frac{1}{2}$, the vertex $t_x$ is reached and player $x$ gets the payoff $0$.

On the other hand, since Solver's expected payoff is at least $1-\epsilon$, from the vertex $(C_i, \neg x)$, player $x$ cannot go to the terminal vertex $t_\solver$ with large probability.

\begin{claim}\label{claim:payoff_solver}
    We have $\sigma_x(C_i, \neg x)(t_\solver) \leq \frac{3 \times 2^m - 2}{2 \times 2^{3m}-1}$.
\end{claim}

\begin{claimproof}
    Call $p$ the probability, when following $\bsigma$ from the vertex $C_1$, that the vertex $t_\solver$ is eventually reached.
    Our constraint imposes $p \leq \epsilon$.
    On the other hand, let $q = \sigma_x(C_i, \neg x)(t_\solver)$.
    The probability $p$ is minimal, for a fixed value of $q$, when we have $i = m$; when from $C_m$, the probability of going to $(C_m, \neg x)$ is exactly $\frac{1}{3}$; and when all other vertices of the form $(C, L)$ that are visited are stochastic vertices.
    In such a case, we have:
    $$p = \frac{1}{2^{m-1}} \left( \frac{1}{3} \left(q + (1-q) \frac{1}{2} p\right) + \frac{2}{3} \frac{1}{2} \frac{1}{2} p \right)$$
    (the probability of reaching $C_m$ is $\frac{1}{2^{m-1}}$; then, with probability $\frac{1}{3}$, we go to the vertex $(C_m, \neg x)$, and with probability $\frac{2}{3}$, to a stochastic vertex).
    By solving this equation, we find:
    $$p = \frac{2q}{3 \times 2^m - 2 + q}.$$
    In the general case, the probability $p$ is therefore greater than or equal to the above quantity.
    Thus, the inequality $p \leq \epsilon$ implies:
    $$\frac{2q}{3 \times 2^m - 2 + q} \leq 2^{-3m},$$
    which implies:
    $$q \leq \frac{3 \times 2^m - 2}{2 \times 2^{3m}-1}$$
    as desired.
\end{claimproof}

    To conclude, given the formula $\phi$, we define the game $\Game^\phi$ as above, and $\epsilon = 2^{-3m}$ (which can be written with polynomially many bits since $m$ is given in unary by the formula).
    We show that, for $m$ large enough, these games constructed are such that the unsatisfiable instances are mapped to the negative instance of the approximate existence problem (no approximate Nash equilibria satisfies the constraints), while 
     the satisfiable formulas, and only them, are mapped to positive instances.
    Therefore, our problem is $\NP$-hard.
\end{proof}

Player $x$ has, therefore, the possibility to deviate: let us consider the strategy $\sigma'_x$, that consists in deterministically going to the vertex $t_\solver$ from the vertex $(C_i, L)$ (and in following $\sigma_x$ on other vertices).
Let us now compute how profitable such a deviation is.

\subp{Minimal profit}
The case where this deviation is the least profitable as the one in which $i = j = m$, in which all other vertices of the form $(C, L)$ that are visited are stochastic vertices, and in which $\sigma_\solver(C_m)(C_m, x) = \sigma_\solver(C_m)(C_m, \neg x) = \frac{1}{3}$.
Let $q = \sigma_x(C_i, \neg x)(t_\solver)$.
Then, the probability of reaching the terminal vertex $t_x$ when following the strategy profile $\bsigma$ is $p$ such that:
$$p = \frac{1}{2^{m-1}} \left(\frac{1}{3} \left(\frac{1}{2} + \frac{1}{2} p\right) + \frac{1}{3} (1 - q) \frac{1}{2} p + \frac{1}{3} \frac{1}{2} \frac{1}{2} p\right)$$
(the probability of reaching $C_m$ is $\frac{1}{2^{m-1}}$; then, with probability $\frac{1}{3}$, we go to $(C_m, x)$; with probability $\frac{1}{3}$, to $(C_m, \neg x)$; and with probability $\frac{1}{3}$, to none of those; and finally, if $s$ is reached, with probability $\frac{1}{2}$ we go back to $C_1$ and have again probability $p$ of reaching $t_s$).
By solving the above equation, we find:
        $$p = \frac{2}{6 \times 2^m + 2q - 5}.$$
Applying \Cref{claim:payoff_solver}, we obtain:
$$p \geq \frac{2}{6 \times 2^m + 2 \frac{3 \times 2^m - 2}{2 \times 2^{3m}-1} - 5}
= \frac{4 \times 2^{3m} - 2}{12 \times 2^{4m} - 10 \times 2^{3m} + 1}.$$

On the other hand, when following the strategy profile $(\bsigma_{-x}, \sigma'_x)$, that probability becomes $p'$ such that:
$$p' = \frac{1}{2^{m-1}} \left(\frac{1}{3} \left(\frac{1}{2} + \frac{1}{2} p'\right) + 0 + \frac{1}{3} \frac{1}{2} \frac{1}{2} p'\right)$$
(same computation, but from the vertex $(C_m, \neg x)$, the probability of reaching $t_x$ is now $0$ because player $x$ goes deterministically to $t_\solver$).
That is:
$$p' = \frac{2}{6 \times 2^m - 3}.$$
Thus, player $x$'s expected payoff is $1-p$ in the first case and $1-p'$ in the second one, and the deviation is profitable by:
$$(1-p') - (1-p) = p-p'  \geq \frac{4 \times 2^{3m} - 2}{12 \times 2^{4m} - 10 \times 2^{3m} + 1} -\frac{2}{6 \times 2^m - 3}$$
$$= \frac{8 \times 2^{3m} - 12 \times 2^m + 4}{(12 \times 2^{4m} - 10 \times 2^{3m} + 1)(6 \times 2^m - 3)}.$$
When $m$ tends to $+\infty$, this quantity is equivalent to $\frac{1}{9} 2^{-2m}$.
Therefore, for large values of $m$, it is greater than $\epsilon = 2^{-3m}$, which means that this deviation is then always (for $m$ large enough) profitable by more than $\epsilon$: the strategy profile $\bsigma$ is not an $\epsilon$-NE, and there is consequently no $\epsilon$-NE where Solver's expected payoff is greater than or equal to $1-\epsilon$.
\end{document}